\documentclass[10pt, journal]{IEEEtran}

\IEEEoverridecommandlockouts  
\usepackage{amsmath}

\usepackage{listings}
\usepackage{indentfirst}
\usepackage{graphicx}
\usepackage{subfigure}
\usepackage{float}
\usepackage{extarrows}
\usepackage{amssymb}
\usepackage{amsthm}
\usepackage{epstopdf}
\usepackage{cite,color}
\usepackage{caption2}
\usepackage{tikz,pgfplots}
\usepackage{xcolor}
\usetikzlibrary{automata,positioning}
\usetikzlibrary{arrows,shapes,chains}
\usepgflibrary{patterns}
\usepackage{algorithm}
\usepackage{algpseudocode}
\usepackage{tikz,pgfplots}
\usetikzlibrary{arrows, intersections}
\pgfplotsset{compat=1.17}

\newtheorem{lemma}{Lemma}

\newtheorem{theorem}{Theorem}
\newtheorem{definition}{Definition}

\newtheorem{proposition}{Proposition}
\newtheorem{remark}{Remark}
\newtheorem{assumption}{Assumption}

\begin{document}
	
\title{\huge Age of Computing: A Metric of Computation Freshness in Communication and Computation Cooperative Networks}
\date{}
\author{Xingran Chen, {\it Member}, IEEE, \IEEEmembership{}
	Yi Zhuang,\IEEEmembership{}
	and Kun Yang, {\it Fellow}, IEEE\IEEEmembership{}
	
	\IEEEcompsocitemizethanks 	 
		{
		\IEEEcompsocthanksitem This work was accepted for publication in the IEEE Transactions on Mobile Computing in June 2025.
		\IEEEcompsocthanksitem This work was supported by Young Scientists Fund of the National Natural Science Foundation of China (Grant No. 62401111), Natural Science Foundation of China (Grant No. 62132004), the Jiangsu Major Project on Basic Researches (Grant No. BK20243059), and Gusu Innovation Project for Talented People (Grant No. ZXL2024360). \textit{(Corresponding Author: Kun Yang)}
		\IEEEcompsocthanksitem Xingran Chen and Yi Zhuang are with School of Information and Communication Engineering,  University of Electronic Science and Technology of China, Chengdu, 611731, China	(E-mail: xingranc@ieee.org, yizhuang265@163.com).
		\IEEEcompsocthanksitem Kun Yang is with the State Key Laboratory of Novel Software Technology, Nanjing University, Nanjing, 210008, China, and School of Intelligent Software and Engineering, Nanjing University (Suzhou Campus), Suzhou, 215163, China, and School of Information and Communication Engineering, University of Electronic Science and Technology of China, Chengdu, China (E-mail: kunyang@nju.edu.cn).
		}
}

\maketitle
\begin{abstract}

In communication and computation cooperative networks (3CNs), timely computation is crucial but not always guaranteed. There is a strong demand for a computational task to be completed within a given deadline. The time taken involves  processing time,  transmission time, and the impact of the deadline. However, a measure of such timeliness in 3CNs is lacking. To address this gap, we propose the novel concept of Age of Computing (AoC) to quantify computation freshness in 3CNs. Built on task timestamps, AoC serves as a practical metric for dynamic and complex real-world 3CNs. 
We evaluate AoC under two types of deadlines:
(i) soft deadline, tasks can be fed back to the source if delayed beyond the deadline, but with additional latency; (ii) hard deadline, tasks delayed beyond the deadline are discarded. We investigate AoC in two distinct networks. In point-to-point, time-continuous networks, tasks are processed sequentially using a first-come, first-served discipline. We derive a general expression for the time-average AoC under both deadlines. Utilizing this expression, we obtain a closed-form solution for M/M/1-M/M/1 systems under soft deadlines and propose an accurate approximation for hard deadlines. These results are further extended to G/G/1-G/G/1 systems. Additionally, we introduce the concept of computation throughput, derive its general expression and an approximation, and explore the trade-off between freshness and throughput. In the multi-source, time-discrete networks, tasks are scheduled for offloading to a computational node. For this scenario, we develop AoC-based Max-Weight policies for real-time scheduling under both deadlines, leveraging a Lyapunov function to minimize its drift.

\end{abstract}

	
\section{Introduction}\label{sec: introduction}
In the 6G era, emerging applications such as the Internet of Things (IoT), smart cities, and cyber-physical systems have significant demands for communication and computation cooperative networks (3CNs), which provide faster data processing, efficient resource utilization, and enhanced security \cite{8016573}. 3CNs originated from mobile edge computing (MEC) technology, which aims to complete computation-intensive and latency-critical tasks, with the paradigm deploying distributedly tons of billions of edge devices at the network edges \cite{7488250}. Besides MEC, 3CNs include fog computing and computing power networks. Fog computing can be regarded as a generalization of MEC, where the definition of edge devices is broader than that in MEC \cite{10.1145/2342509.2342513}. Computing power networks refer to a broader concept of distributed computing networks, including edge, fog, and cloud computing \cite{9354741}. In all 3CNs, there is no established metric for capturing the freshness of computation. Recently, a notable metric called the Age of Information (AoI) has been proposed to describe information freshness in communication networks \cite{AoItutorial}. The AoI metric has broad applications in various communication and control contexts, including random access protocols \cite{cxrAoI}, multiaccess protocols \cite{PAoI2015}, remote estimation \cite{cxrestimation}, wireless-powered relay-aided communication networks \cite{KunyangAoI-1}, and network coding \cite{cxrisit, cxritw}.

However, applying AoI in 3CNs is inappropriate because it only addresses communication latency and does not account for computation latency. In this paper, we propose a novel metric called the {\it Age of Computing (AoC)} to capture computation freshness in 3CNs. A primary requirement in 3CNs is that computational tasks are processed as promptly as possible and within a maximum acceptable deadline. The core idea of AoC is to combine communication delay, computation delay, and the impact of the maximum acceptable deadline. Communication and computation delays are caused by the transmission and processing of computational tasks, while the impact of the maximum acceptable deadline accounts for additional delays when task delays exceed the users' acceptable threshold.

\subsection{Related Work}\label{subsec: Related Work}
All related papers can be divided into two broad categories. The first category investigates information freshness in edge and fog computing networks. The second category focuses on freshness-oriented metrics.

\subsubsection{Information Freshness in Edge/Fog Computing Networks}\label{subsubsec: Information Freshness in Edge/Fog Computing Networks}
In edge and fog computing networks, tasks or messages typically go through two phases: the transmission phase and the processing phase. The basic mathematical model for these networks is established as two-hop networks and tandem queues.

The first study to focus on AoI for edge computing applications is \cite{AoIedge1}, which primarily calculated the average AoI. As an early work, \cite{MPeakAoI} established an analytical framework for the peak age of information (PAoI), modeling the computing and transmission process as a tandem queue. The authors derived closed-form expressions and proposed a derivative-free algorithm to minimize the maximum PAoI in networks with multiple sensors and a single destination. Subsequently, \cite{PAoIedge} modeled the communication and computation delays as a generic tandem of two first-come, first-serve (FCFS) queues, and analytically derived closed-form expressions for PAoI in M/M/1-M/D/1 and M/M/1-M/M/1 tandems.  More recently, \cite{kam2022age} analyzed the average AoI for a tandem of FCFS M/M/1-M/M/1 queues with non-preemptive, memoryless servers. Their study considered both single-capacity and infinite-capacity queue settings, and utilized stochastic hybrid system techniques to obtain the results.  Additionally, the authors in \cite{Jacksonnet} extended the concept of tandems of M/M/1 queues, examining the AoI in Jackson Networks with finite buffer sizes. They provided a closed-form expression for an upper bound on the average AoI, also utilizing stochastic hybrid systems in their analysis. Building on \cite{PAoIedge, kam2022age, Jacksonnet}, \cite{ComputationTransmission} went further by considering both average and peak AoI in general tandems with packet management. The packet management included two forms: no data buffer, and a one-unit data buffer with last-come first-serve discipline. This work illustrated how computation and transmission times could be traded off to optimize AoI, revealing a tradeoff between average AoI and average peak AoI.  Expanding on \cite{ComputationTransmission}, \cite{AoI_edgecomputing} explored the information freshness of Gauss-Markov processes, defined as the process-related timeliness of information. The authors derived closed-form expressions for information timeliness at both the edge and fog tiers. These analytical formulas explicitly characterize the dependency among task generation, transmission, and execution, serving as objective functions for system optimization. 
Reference \cite{ZChen-2} investigated a remote status updating system where the transmission process is modeled as a multi-stream M/G/1/1 non-preemptive system. The study derived closed-form expressions for the average AoI of each stream in a heterogeneous case, where the service time distributions differ across streams. In \cite{AoIedgecomputing}, a multi-user MEC network where a base station (BS) transmits packets to user equipment was investigated. The study derived the average AoI for two computing schemes—local computing and edge computing—under a first-come, first-serve discipline. 

There are other relevant works such as \cite{NOMAedge, ChenHe2021, FRANcxr, decentralizedMEC, RLEdge, ZChen-1}. Reference \cite{NOMAedge} investigated information freshness in MEC networks from a multi-access perspective, where multiple devices use NOMA to offload their computational tasks to an access point integrated with an MEC server. Leveraging tools from queuing theory, the authors proposed an iterative algorithm to obtain the closed-form solution for AoI.  The authors  in \cite{ChenHe2021} investigated the age minimization problem in a two-hop relay system, subject to a constraint on the average number of forwarding operations at the relay. They proposed a low-complexity double threshold relaying policy and derived approximate closed-form expressions for the average AoI at the destination and the average number of forwarding operations at the relay.
In \cite{FRANcxr}, a F-RAN with multiple senders, multiple relay nodes, and multiple receivers was considered. The authors analyzed the AoI performances and proposed optimal oblivious and non-oblivious policies to minimize the time-average AoI.
\cite{decentralizedMEC} and \cite{RLEdge} explored AoI performances in MEC networks using different mathematical tools. \cite{decentralizedMEC} considered MEC-enabled IoT networks with multiple source-destination pairs and heterogeneous edge servers. Using game-theoretical analysis, they proposed an age-optimal computation-intensive update scheduling strategy based on Nash equilibrium. 
Reinforcement learning is also a powerful tool in this context. \cite{RLEdge} proposed a computation offloading method based on a directed acyclic graph task model, which models task dependencies. The algorithm combined the advantages of deep Q-network, double deep Q-network, and dueling deep Q-network algorithms to optimize AoI.  Parallel queues were explored in \cite{ZChen-1}, where the average age and PAoI in a dual-queue status update system, monitoring a common stochastic process through two independent channels, were analyzed. Closed-form expressions were derived using the graphic method and state flow graph analysis.

\subsubsection{Freshness-oriented Metrics}\label{subsubsec: Freshness-oriented Metrics}
The AoI metric, introduced in \cite{AoItutorial}, measures the freshness of information at the receiver side. AoI depends on both the frequency of packet transmissions and the delay experienced by packets in the communication networks \cite{cxrAoI}. When the communication rate is low, the receiver’s AoI increases, indicating stale information due to infrequent packet transmissions. However, even with frequent transmissions, if the system design imposes significant delays, the receiver's information will still be stale.
Following the introduction of AoI, several related metrics were proposed to capture network freshness from different perspectives. Peak AoI, introduced in \cite{PAoI2015}, represents the worst-case AoI. It is defined as the maximum time elapsed since the preceding piece of information was generated, offering a simpler and more mathematically tractable formulation.

Nearly simultaneously, the age of synchronization (AoS) \cite{AoS} and the effective age \cite{Effectiveage} were proposed. AoS, as a complementary metric to AoI, drops to zero when the transmitter has no packets to send and grows linearly with time until a new packet is generated \cite{AoS}. The effective age metrics in \cite{Effectiveage} include sampling age, tracking the age of samples relative to ideal sampling times, and cumulative marginal error, tracking the total error from the reception of the latest sample to the current time.

Later, the age of incorrect information (AoII) \cite{AoII2020} and the urgency of information (UoI) \cite{UoI2020} were introduced. AoII addresses the shortcomings of both AoI and conventional error penalty functions by extending the concept of fresh updates to ``informative'' updates—those that bring new and correct information to the monitor side \cite{AoII2020}. UoI, a context-based metric, evaluates the timeliness of status updates by incorporating time-varying context information and dynamic status evolution \cite{UoI2020}, which enables analysis of context-based adaptive status update schemes and more effective remote monitoring and control.

Despite the variety of freshness-oriented metrics proposed, none are applicable for capturing computation freshness in 3CNs. None of these metrics simultaneously address the impact of both communication and computation delays, as well as the maximum acceptable deadline. Motivated by the need for a metric capturing freshness in 3CNs, we propose the AoC metric in this paper.

\subsection{Contributions}\label{subsec: Contributions}

This paper introduces the Age of Computing (AoC), a novel metric designed to capture computation freshness in 3CNs (see Definition~\ref{def: AoC}). The AoC concept is built on tasks' arrival and completion timestamps, which makes it applicable to dynamic and complex real-world 3CNs. The AoC is defined under two types of deadlines:  (i) soft deadline, i.e., if the task's delay exceeds the deadline, the outcome is still usable but incurs additional latency; (ii) hard deadline, i.e., if the task's delay exceeds the deadline, the outcome is discarded, and the task is considered invalid.

We then theoretically analyze the time-average AoC under both types of deadlines in a linear topology comprising a source, a transmitter, a receiver, and a computational node. Tasks arrive at the source at a constant rate and immediately enter a communication queue at the transmitter. After being transmitted/offloaded to the receiver, tasks are forwarded to the computational node for processing in a computation queue. The queuing discipline considered is first-come, first-served.

Under the soft deadline, we first derive a general expression for the average AoC (see Theorem~\ref{thm: average AoC general}). We then study a fundamental scenario where the task arrival process follows a Poisson distribution, and the transmission and computation delays adhere to exponential distributions—forming an M/M/1-M/M/1 system. In this case, we derive a closed-form expression for the average AoC (see Theorem~\ref{thm: average AoC soft}). Subsequently, we extend our analysis to a more general scenario where the task arrival process follows a general distribution, and the transmission and computation delays also follow general distributions—resulting in a G/G/1-G/G/1 system. For this case, we derive the expression for the average AoC as well (see Theorem~\ref{thm: average AoC soft extension}).

Under the hard deadline, we first derive a general expression for the average AoC (see Theorem~\ref{thm: average AoC hard general}). This expression involves intricate correlations, making it highly challenging to obtain a closed-form solution, even for an M/M/1-M/M/1 system. To address this, we provide an approximation of the average AoC and demonstrate its accuracy when the communication and computation rates are significantly larger than the task generation rate (see Theorem~\ref{thm: average AoC hard}). We also define computation throughput as the number of tasks successfully fed back to the source per time slot (see Definition~\ref{def: computation throughput}). A general expression for computation throughput is derived (see Lemma~\ref{lem: computation throughput}), along with an approximation (see Proposition~\ref{pro: approximated CT}). Furthermore, we explore the trade-off between computation freshness and computation throughput (see Lemma~\ref{lem: pareto optimal}).  In the end, we extend the accurate approximations for both the average AoC and computation throughput to more generalized scenario involving G/G/1-G/G/1 systems (see Theorem~\ref{thm: average AoC hard extension} and Proposition~\ref{pro: approximated CT extension}).

Finally, we apply the AoC concept to develop optimal real-time scheduling strategies focused on enhancing computation freshness in multi-source networks. Recognizing the importance of recent computational tasks, we adopt a preemptive scheduling rule. For real-time scenarios, we propose AoC-based Max-Weight policies for both deadlines (see Algorithms~\ref{alg: soft} and~\ref{alg: hard}). By constructing a Lyapunov function \big(see \eqref{eq: Lyapunov function}\big) and its drift \big(see \eqref{eq: Drift}\big), we show that these policies minimize the Lyapunov drift in each time slot (see Propositions~\ref{pro: maxweight soft} and~\ref{pro: maxweight hard}).

The remaining parts of this paper are organized as follows. Section~\ref{sec: Age of Computing} proposes and discusses the novel concept AoC. Section~\ref{sec: AoC under Soft Deadlines} and Section~\ref{sec: AoC under Hard Deadlines} derive theoretical results for the AoC under the soft and hard deadlines, respectively. Section~\ref{sec: AoC-based Scheduling for Multi-Source 3CNs} proposes real-time AoC-based scheduling algorithms for multi-source networks. We numerically verify our theoretical results in Section~\ref{sec: simulations} and conclude this work in Section~\ref{sec: conclusion}.

\section{Age of Computing}\label{sec: Age of Computing}
In this section, we introduce the mathematical formulation of the novel concept, Age of Computing (AoC), which quantifies the freshness of computation within 3CNs. Consider a line topology comprising a source, a transmitter, a receiver, and a computational node (sink), as depicted in Fig.~\ref{outline}. In this topology, both the transmitter and the computational node are equipped with queuing capabilities. The process begins with the source generating or offloading computational tasks, which are first placed into a communication queue at the transmitter, awaiting transmission. Once transmitted, the tasks are received by the receiver, and then forwarded to the computational node. At the computational node, the tasks enter a computation queue, where they await processing before eventually departing from the system.

\begin{figure}[htbp]
\centering
\includegraphics[width=0.45\textwidth,height=0.15\textwidth]{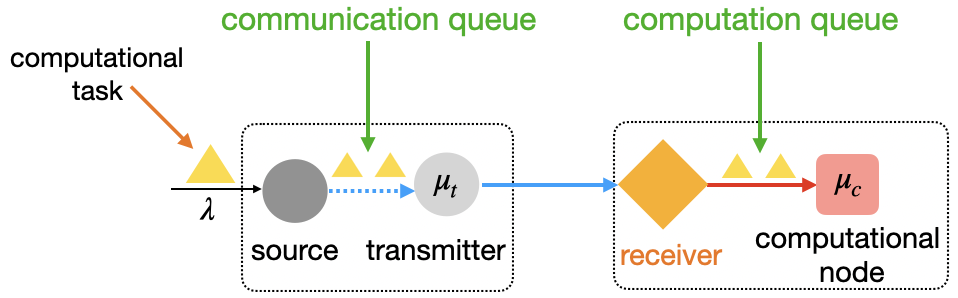}
\caption{A line topology consisting of a source, a transmitter, a receiver, and a computational node (sink).}
\label{outline}
\end{figure}

\subsection{Definition}\label{subsec: definition}

In the network, the queuing discipline follows a {\it first-come first-serve} approach.
For any task $k$, let $\tau_k$ denote the arrival time at the source, $\tau_k''$ the time when the computation starts at the computational node, and $\tau_k'$ the time when the computation completes. The {\it delay} of task $k$ is then defined as $\tau_k'-\tau_k$. A task $k$ is considered {\it valid} if its outcome can be fed back to the source\footnote{The feedback is an acknowledgment message with a small bit size, commonly used in communication systems to confirm the successful receipt of data.}; otherwise, it is deemed invalid.

\begin{definition}\label{def: informative task}
(informative, processing, and latest tasks). The index of the informative task during $[0, t]$, denoted by $N(t)$, is given by 
\begin{align}\label{eq: index of latest packet}
N(t) = \max\{k|\tau_k'\leq t, \text{and task }k\text{ is valid}\}.
\end{align}
The index of the processing task during $[0, t]$, denoted by $P(t)$, is given by 
\begin{align}\label{eq: index of processing packet}
P(t) = \max\{k|\tau_k''\leq t\}.
\end{align}
The index of the latest task during $[0, t]$, denoted by $G(t)$, is given by 
\begin{align}\label{eq: index of information packet}
G(t) = \max\{k|\tau_k'\leq t\}.
\end{align}
\end{definition}

Based on Definition~\ref{def: informative task}, an informative task is a valid task that brings the latest information. The processing is the current task being processed, and the latest refers to the last completed task.  The informative task and the latest task are not necessarily the same, i.e., $G(t)\geq N(t)$. They coincide \big($G(t)=N(t)$\big) only when the latest task is also informative.  At any time $t$, if the computational node is idle (no task is being processed), then the processing task is exactly the latest one, i.e, $P(t)=G(t)$. If the computational node is occupied (a task is being processed), then $P(t)=G(t)+1$. In summary, we have $P(t)\geq G(t)\geq N(t)$.  A graphical illustration of Definition~\ref{def: informative task} is shown in Fig.~\ref{PGN}, which depicts the system state at time $t$. In this example, task~$k+2$ arrives at the computational node at time~$t$, while tasks~$k$ and~$k+1$ have already departed from the computational node. Among them, task~$k$ is valid, whereas task~$k+1$ is invalid. According to Definition~\ref{def: informative task}, this yields $P(t) = k+2$, $G(t) = k+1$, and $N(t) = k$, satisfying the condition $P(t) \geq G(t) \geq N(t)$.

\begin{figure}[htbp]
	\centering
	\includegraphics[width=0.45\textwidth,height=0.21\textwidth]{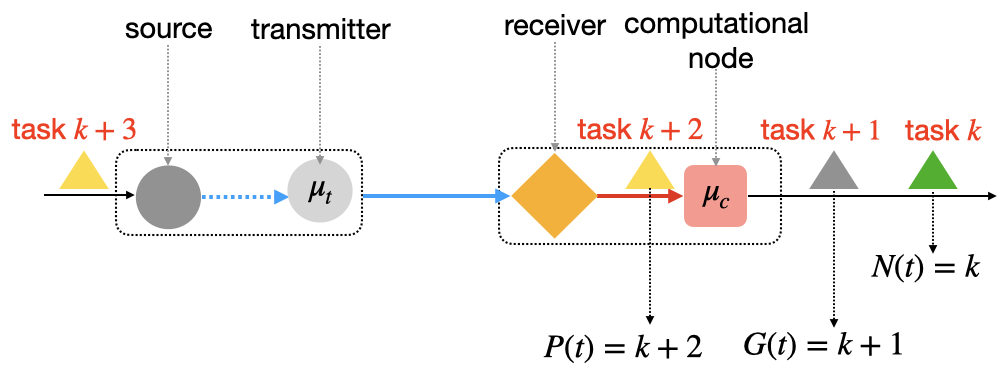}
	\caption{A graphical illustration for Definition~\ref{def: informative task}: At time $t$, task $k+2$ arrives at the computational node, while tasks $k$ and $k+1$ already departed from the computational node. Among them, task $k$ (green) is valid, whereas task $k+1$ (gray) is invalid.}
	\label{PGN}
\end{figure}

\begin{definition}\label{def: softhard deadline}
A maximum acceptable deadline $w>0$ can be categorized into two types:
\begin{itemize}
	\item Soft deadline: A task is considered valid if its delay $\tau_k'-\tau_k > w$.
	\item Hard deadline: A task is considered invalid if its delay $\tau_k'-\tau_k > w$.
\end{itemize}
\end{definition}
We define the number of tasks within $[0, t]$ whose delays exceed the threshold $w$ as:
\begin{align}\label{eq: index of delay not exceeding w}
A(t) = \operatorname{card}\Big(\big\{1\leq k\leq G(t): \tau_k' -\tau_k>w \big\}\Big).
\end{align}
In \eqref{eq: index of delay not exceeding w}, $\operatorname{card}(\cdot)$ denotes the cardinality. This quantity can be efficiently tracked by the computational node using mechanisms such as a local counter or a task log. The computation freshness at the computational node is formally defined as follows.
\begin{definition}\label{def: AoC}
(AoC). Under the soft deadline, the age of computing (AoC) is defined as the random process
\begin{align}\label{eq: AoC-soft}
c_{\text{soft}}(t) =& t - \tau_{N(t)}\nonumber\\
+&1_{\{P(t)>G(t)\}}\cdot\frac{A(t)}{G(t)}\cdot(t - \tau_{P(t)}-w)^+.
\end{align}
Under the hard deadline, the AoC is defined as the random process
\begin{align}\label{eq: AoC-hard}
c_{\text{hard}}(t) = t-\tau_{N(t)}.
\end{align}
 \end{definition}

The key distinction in Definition~\ref{def: AoC} lies in how computation freshness is assessed under hard and soft deadlines:
\begin{itemize}
	\item Under the hard deadline, computation freshness in \eqref{eq: AoC-hard} is determined solely by the informative tasks, as tasks with delays exceeding the threshold $w$ are deemed invalid. 
	\item Under the soft deadline, computation freshness in \eqref{eq: AoC-soft} accounts for both the informative tasks (i.e, $t-\tau_{N(t)}$) and an additional latency incurred if the task being processed experiences a significant instantaneous delay \big(i.e, $1_{\{P(t)>G(t)\}}\frac{A(t)}{G(t)}(t-\tau_{P(t)}-w)^+$\big). In particular, the additional latency includes $3$ components:
\begin{itemize}
	\item The indicator function $1_{\{P(t)>G(t)\}}$ denotes whether a task is being processed at the computational node.
	\item  The ratio
	$\frac{A(t)}{G(t)}$  represents the frequency of task delays exceeding the deadline up to time $t$, effectively quantifying {\it the level/frequency of conflict} with respect to the deadline (up to time $t$). 
	\item  The term $(t-\tau_{P(t)}-w)^+$ quantifies the amount by which the delay of the task currently being processed exceeds the threshold.
	\item This additional latency, calculated by the computational node, vanishes as soon as the current task is completed.
\end{itemize}
\end{itemize}

The concept of information freshness, known as AoI \cite{AoItutorial}, reflects the cumulative delay over a given time period. Building on this foundation, AoC extends the notion to computational tasks, representing the cumulative delay associated with their processing. AoC provides a more comprehensive understanding of the timeliness of computations in a system.
While related, AoC and AoI are {\it fundamentally distinct} in their definitions and physical interpretations: AoC evaluates the freshness of the informative task (may includ an additional latency incurred by the task being processed), whereas AoI focuses on the freshness of the latest task.

Since the AoC $c_{x}(t)$ with $x\in\{\text{soft}, \text{hard}\}$ captures the computation freshness at time $t$, we often consider the time-average AoC over a period to measure the computation freshness of a network. As $T\to\infty$, we define the average AoC of a network as 
\begin{align}\label{eq: average AoC-1}
\Theta_{x} \triangleq  \lim_{T\to\infty}\frac{1}{T}\int_0^T c_{x}(t) dt,\quad x\in\{\text{soft}, \text{hard}\}.
\end{align}

Finally, we summarize the important notations and their descriptions in the paper in Table~\ref{tab: notations and descriptions}.

\begin{table}[ht]
	\centering
	\begin{tabular}{|c|c|}
		\hline
		 $k$ & the index of computational tasks \\
		\hline
		$w$ & the threshold \\
				\hline
		$\tau_k$ & the arrival time of task~$k$ at the source\\
				\hline
		$\tau_k''$ &the time  when the computation starts \\
				\hline
		$\tau_k'$ & the time when the computation completes \\
				\hline
		$N(t)$ & the index of the informative task during $[0, t]$\\
				\hline
		$P(t)$ & the index of the processing task during $[0, t]$\\
				\hline
		$G(t)$ & the index of the latest task during $[0, t]$\\
				\hline
		$A(t)$ & the number of tasks with delays $> w$  during $[0, t]$\\
				\hline
		$c_{\text{soft}}(t)$ & the AoC under the soft deadline in time $t$\\
				\hline
		$c_{\text{hard}}(t)$ & the AoC under the hard deadline in time $t$\\
				\hline
		$\Theta_{\text{soft}}$ & the time-average AoC under the soft deadline \\
				\hline
		$\Theta_{\text{hard}}$ & the time-average AoC under the hard deadline \\
				\hline
		$\mu_t$ & the communication rate\\
				\hline
		$\mu_c$ & the computation rate\\
				\hline
		$X_{k+1}$ & the inter-arrival time between task $k$ and task $k+1$\\
				\hline
		$T_{k}$ & the delay of task $k$\\
				\hline
		$M$ & the number of invalid tasks between two valid tasks\\
				\hline
	$S_{k, t}$ & the transmission delay of task $k$ \\ 
			\hline
	$S_{k, c}$ & the computation delay of task  $k$ \\ 
		\hline
	$ \hat{\epsilon}(t)$ & $A(t)/G(t)$ \\ 
			\hline
	$\epsilon_w$ & $\Pr(T_k>w)$ \\ 
		\hline
	\end{tabular}
	\caption{Summary of Important Notations}
	\label{tab: notations and descriptions}
\end{table}

\subsection{Insights and Applications}\label{subsec: Insights and Applications}

From \eqref{eq: index of delay not exceeding w}, \eqref{eq: AoC-soft} and \eqref{eq: AoC-hard}, we observe that the {\it only} notation used is the {\it timestamp}, specifically the arrival timestamp $\tau_k$ and the departure timestamp $\tau_k'$. To establish this concept, we intentionally avoid incorporating physical parameters such as bandwidth, GPU cycles, or caching. Instead, we focus solely on the use of timestamps. The reason for this choice is that the timestamp is the most fundamental and universally applicable metric that can be recorded by the sink (i.e., computational nodes) in realistic 3CNs.

\begin{figure}[htbp]
	\centering
	\includegraphics[width=0.45\textwidth,height=0.25\textwidth]{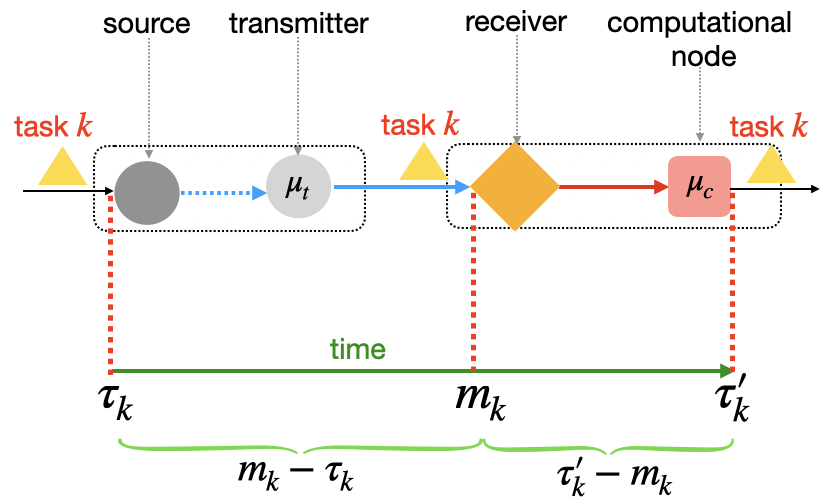}
	\caption{A graphical illustration that conveys the intuition underlying Definition~\ref{def: AoC}.}
	\label{insights}
\end{figure}

Let $m_k$ denote the timestamp when task $k$ reaches the receiver  (see Fig.~\ref{insights}). The difference $m_k - \tau_k$ represents the system time incurred during communication for task $k$, which is determined by the communication capability. This delay can be influenced by various network characteristics, such as bandwidth, information size, signal-to-noise ratio, channel capacity, network load, transmission protocols, and more. On the other hand, the difference $\tau_k' - m_k$ represents the system time incurred during computation for task $k$, which depends on the computation capability. This delay is influenced by factors such as GPU cycles, task complexity, node load, memory size, access speed, I/O scheduling, and more. Finally,  the total delay of task $k$, given by $\tau_k' - \tau_k$, consists of the system times incurred during both communication and computation phases. As such, it is influenced by both the communication and computation capabilities. The AoC concept effectively captures the impact of both these factors on task performance.

The AoC concept is also applicable in dynamic and complex environments, such as mobile computing. In such settings, the communication network may experience interruptions or temporary failures due to factors like signal loss, interference, or mobility-related issues (e.g., moving out of network coverage). Tasks may experience temporary delays in a queue due to these disruptions. Nevertheless, the AoC concept remains applicable. To calculate the AoC, we only need the timestamps of tasks (i.e., $(\tau_k, \tau_k')$). As long as timestamps can be recorded, the AoC can be calculated, regardless of disruptions.

\section{AoC Analysis under the Soft Deadline}\label{sec: AoC under Soft Deadlines}

In this section, we first provide graphical insights into the curve of the AoC (see Fig.~\ref{curveAoCsoft}). Next, we derive a general expression for the time-average AoC, presented in Theorem~\ref{thm: average AoC general} (Section~\ref{subsec: General Expression for soft}). Following this, we analyze the expression in a fundamental case: the M/M/1 - M/M/1 tandem system, as detailed in Theorem~\ref{thm: average AoC soft} (Section~\ref{subsec: Average AoC in M/M/1 - M/M/1 Tandem}). Finally, we extend the average AoC expression to more general cases, specifically G/G/1 - G/G/1 tandem systems, as discussed in Theorem~\ref{thm: average AoC soft extension} (Section~\ref{subsec: Extension soft}).

\subsection{General Expression for $\Theta_{\text{soft}}$}\label{subsec: General Expression for soft}

For analytical tractability, we consider a simplified model: (i) A task is represented by $(L, w, B)$, where $L$ denotes the task's input data size, $w$ represents the maximum acceptable deadline, and $B$ indicates the computation workload \cite{8016573}. (ii) The network is characterized by $(R, F)$, where $R$ is the data rate of the communication channel at the transmitter, and $F$ is the CPU cycle frequency at the computational node. (iii) The expected transmission delay of a task is $L/R$, and the expected computation delay is $B/F$. (iv) Both delays each follow their respective distributions. We define $\mu_t = R/L$ and $\mu_c = F/B$ as the {\it communication rate} and {\it computation rate}, respectively.

Let $h(t)=t-\tau_{N(t)}$ and $\hat{\epsilon}(t)=\frac{A(t)}{G(t)}$.  From \cite{AoItutorial}, the AoI concept shares the same formula as $h(t)$. Using the AoI curve as a benchmark, the AoC curve is depicted in Fig.~\ref{curveAoCsoft}. In Fig.~\ref{curveAoCsoft}~(a), the delay of task $k$ exceeds the threshold ($T_k>w$), but any waiting time before task $k$ is processed does not surpass the threshold.  When the (instantaneous) delay of a task is less than the deadline $w$, the AoC curve coincides with the AoI curve. However, when the (instantaneous) delay of a task exceeds the deadline $w$, the portion of the delay exceeding the deadline increases at a rate $1+\hat{\epsilon}(t)$. In Fig.~\ref{curveAoCsoft}~(b), both the delay of task $k$ and the waiting time before  task  $k$ is processed  ($\tau_{k-1}'-\tau_k>w$) exceed the threshold. When the processing of task $k$ begins, an additional latency is introduced. As a result, at time  $\tau_{k-1}'$, there is an upward jump (additional latency) of $\big(1+\hat{\epsilon}(t)\big)(\tau_{k-1}'-\tau_k-w)$.  Subsequently, the AoC curve increases at a rate $1+\hat{\epsilon}(t)$. 

\begin{figure}[htbp]
	\centering
	\subfigure[The waiting time for task $k$ in the computation queue does not exceed $w$.] {\includegraphics[width=0.44\textwidth,height=0.36\textwidth]{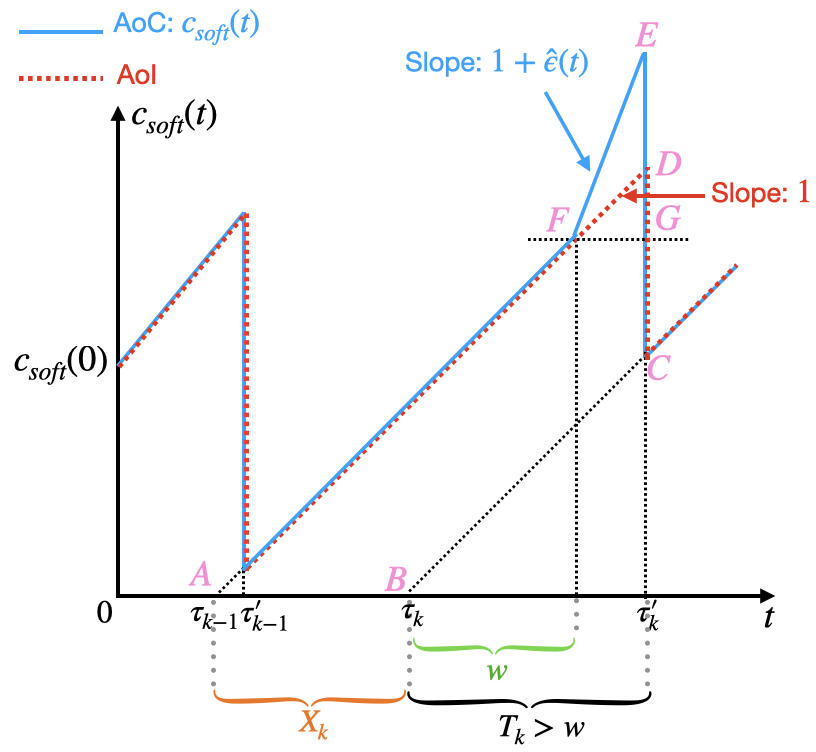}}
	\subfigure[The waiting time for task $k$ in the computation queue exceeds $w$.] {\includegraphics[width=0.5\textwidth,height=0.36\textwidth]{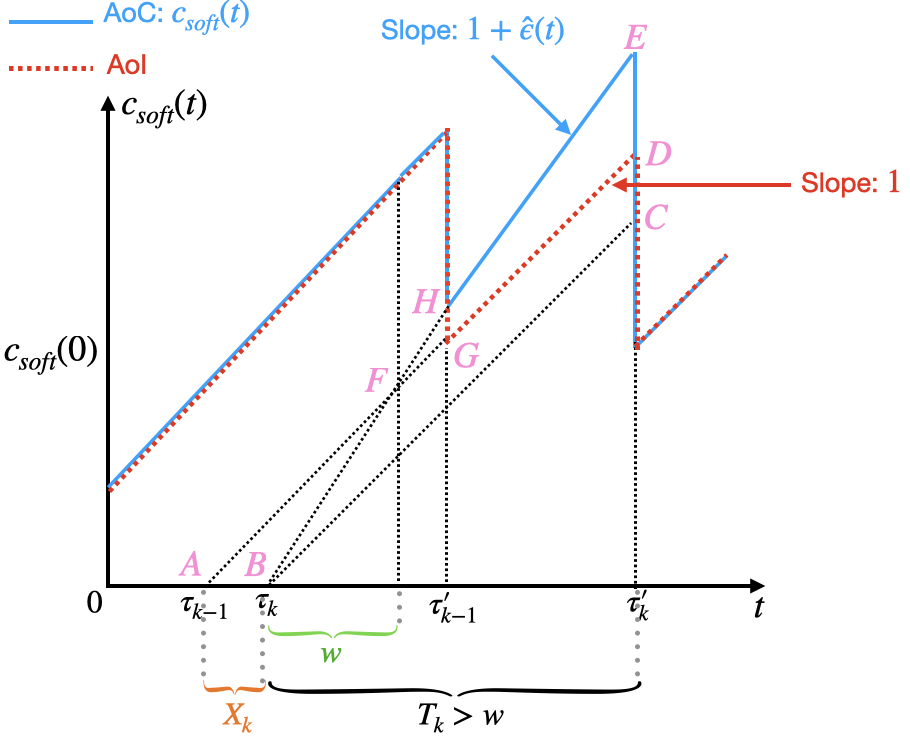}}
	\caption{The curve AoC under the soft deadline}
	\label{curveAoCsoft}
\end{figure}

To uncover theoretical insights into the average AoC, we investigate it within queue-theoretic systems. We assume a stationary queuing discipline with a first-come, first-serve approach.  On the source side, computational tasks arrive according to a stationary stochastic process characterized by a constant average rate. The inter-arrival time between consecutive tasks is denoted by $X_{k+1} = \tau_{k+1} - \tau_k$.  Upon arrival, each task experiences a transmission delay, denoted by $S_{k, t}$, at the transmitter, which follows a distribution with a constant expectation. Similarly, the computation delay denoted by $S_{k, c}$, at the computational node also follows a distribution with a constant expectation. Let the delay of task $k$ be denoted by $T_k$. In queuing systems, this delay is referred to as the system time, and we use these terms interchangeably. Under the stationary queuing discipline, $\{X_k\}_k$, $\{S_{k, t}\}_k$, $\{S_{k, c}\}_k$, and $\{T_k\}_k$ are i.i.d, respectively.

\begin{theorem}\label{thm: average AoC general}
	The average AoC can be calculated as
	\begin{align}\label{eq: average AoC soft}
		&\Theta_{\text{soft}} =\frac{\mathbb{E}[X_kT_k]+\frac{1}{2}\mathbb{E}[X_k^2]}{\mathbb{E}[X_k]}\nonumber\\
		&+\epsilon_w\cdot\frac{\mathbb{E}\big[\big((T_k-w)^+\big)^2\big] - \mathbb{E}\big[\big((T_k-S_{k,c}-w)^+\big)^2\big]}{2\mathbb{E}[X_k]},
	\end{align}
	where $\epsilon_{w}=\Pr(T_k>w)$.
\end{theorem}
\begin{proof}
The proof is given in Appendix~\ref{App: average AoC soft}.
\end{proof}

\subsection{Average AoC in M/M/1-M/M/1 Systems}\label{subsec: Average AoC in M/M/1 - M/M/1 Tandem}
In this section, we let the inter-arrival time $X_k$ follows an exponential distribution with parameter $\lambda$, meaning that computational tasks arrive at the source according to a Poisson process characterized by an average rate of $\lambda$. Let $S_{k, t}$ and $S_{k, c}$ have exponential distributions with parameters $\mu_t$ ($=R/L$) and $\mu_c$ ($=F/X$), respectively. In other words, the network forms an M/M/1-M/M/1 tandem. 

\begin{theorem}\label{thm: average AoC soft}
Let $\rho_t=\lambda/\mu_t$ and $\rho_c=\lambda/\mu_c$. Denote
\begin{align*}
&\delta_t=1-\rho_t,\,\delta_c=1-\rho_c,\, \zeta_t=e^{-\mu_t\delta_tw},\,\zeta_c=e^{-\mu_c\delta_cw}.
\end{align*} 
The closed form expression for $\Theta_{\text{soft}}$ is given by: if $\mu_t\neq \mu_c$, then
\begin{align}\label{eq: AoC soft-1}
	\Theta_{\text{soft}} = &\frac{1}{\lambda}+\frac{1}{\mu_t}+\frac{1}{\mu_c}+\frac{\rho_t^2}{\mu_t \delta_t}+ \frac{\rho_c^2}{\mu_c\delta_c}+\frac{\rho_t\rho_c}{\mu_t+\mu_c-\lambda}\nonumber\\
	+&\lambda\frac{\mu_c(\delta_c\mu_t\delta_t)^2}{(\mu_c-\mu_t)^2}\big(\frac{\zeta_t}{\mu_t\delta_t} - \frac{\zeta_c}{\mu_c\delta_c}\big)\big(\frac{\zeta_t}{\mu_t^2\delta_t^2} - \frac{\zeta_c}{\mu_c^2\delta_c^2}\big);
\end{align}
if $\mu_t=\mu_c$, then
\begin{align}\label{eq: AoC soft-2}
\Theta_{\text{soft}} = &\frac{1}{\lambda}+\frac{2}{\mu_t}+\frac{2\rho_t^2}{\mu_t\delta_t}+\frac{\rho_t^2}{\mu_t+\mu_t\delta_t}\nonumber\\
+&\lambda\zeta_t^2(1+\mu_t\delta_tw)(\frac{2}{\mu_t^2\delta_t}+\frac{w}{\mu_t}).
\end{align}
\end{theorem}
\begin{proof}
The proof is given in Appendix~\ref{App: average AoC soft closedform}.
\end{proof}

\subsection{Extension: Average AoC in G/G/1-G/G/1 Systems}\label{subsec: Extension soft}
In this section, we let the inter-arrival time $X_k$ follows a general distribution with $\mathbb{E}[X_k] = 1/\lambda$, meaning that computational tasks arrive at the source according to a general stochastic process characterized by an average rate of $\lambda$. Let $S_{k, t}$ and $S_{k, c}$ have general distributions with $\mathbb{E}[S_{k, t}]=1/\mu_t$ and $\mathbb{E}[S_{k, c}] = 1/\mu_c$, respectively. In other words, the network forms a G/G/1-G/G/1 tandem. 

To obtain the average AoC, it is necessary to know the stochastic information about the inter-arrival interval, the service delay at the transmission and computation nodes, and the system time of a task in both the transmission and computation queues. Let the system time of a task in the transmission and computation queues be denoted as $U_{k, t}$ and $U_{k, c}$, respectively. 

Note that sequences $\{X_k\}_k$, $\{S_{k, t}\}_k$, $\{S_{k, c}\}_k$, $\{U_{k, t}\}_k$, and $\{U_{k, c}\}_k$ are i.i.d with respect to $k$, respectively. The correpsonding density functions are denoted by $f_X$, $f_{S_t}$, $f_{S_c}$, $f_{U_t}$, and $f_{U_c}$, respectively. 
 The joint density function of $U_{k, t}$ and $U_{k, c}$ is denoted as $f_{U_t, U_c}$. 
The following assumptions are then required. 

\begin{assumption}\label{ass: distributions}
The density functions $f_X$, $f_{S_t}$, $f_{S_c}$, and $f_{U_t, U_c}$ are known.
\end{assumption}
Based on Assumption~\ref{ass: distributions}, the following steps can be taken:
\begin{itemize}
	\item By integrating $f_{U_t, U_c}$, we can obtain the marginal density functions $f_{U_t}$ and $f_{U_c}$.
	\item Since $S_{k, t}$ (respectively, $S_{k, c}$) is independent of $U_{k, t}-S_{k, t}$ (respectively, $U_{k, c}-S_{k, c}$), and the marginal density functions $f_{U_t}$ and $f_{U_c}$ are available,  we can derive the density functions of the waiting times in both queues, namely $f_{U_t-S_t}$ and $f_{U_c-S_c}$, through inverse convolution.
	\item Given that $f_{U_t-S_t}$,  $f_{U_c-S_c}$, and $f_{U_t, U_c}$ are known, we can again apply inverse convolution to derive the joint density functions $f_{U_t, U_c-S_c}$ and $f_{U_c, U_t-S_t}$. 
\end{itemize}

\begin{theorem}\label{thm: average AoC soft extension}
When Assumption~\ref{ass: distributions} is satisfied, let $\mu_t$ and $\mu_c$ be in Theorem~\ref{thm: average AoC soft}. The closed form expression for $\Theta_{\text{soft}}$ is given by:
\begin{align}\label{eq: average AoC soft extension}
\Theta_{\text{soft}} =& \frac{1}{\mu_t}+\frac{1}{\mu_c} + \lambda\big(\frac{\int_{0}^{\infty}x^2f_X(x)dx}{2}+g_1+ g_2\big)\nonumber\\
+&\lambda \frac{\int_{w}^{\infty}\eta_1(\tau)d\tau}{2}\int_{w}^{\infty}(\tau-w)^2(\eta_1(\tau)-\eta_2(\tau))d\tau\big),
\end{align}
where  $\eta_1(\tau) = \int_{0}^{\tau}f_{U_t, U_c}(u, \tau-u)du$, $\eta_2(\tau) = \int_{0}^{\tau}f_{U_t, U_c-S_c}(u, \tau-u)du$,
\begin{align*}
g_1 = &\int_{0}^{\infty} x \int_{x}^{\infty} (\tau-x)f_{U_t}(\tau)d\tau f_X(x)dx,\\
g_2 =& \int_{0}^{\infty}x \int_{0}^{\infty} \tau 
\int_{0}^{\infty} (\xi(x)\cdot f_{S_{t}}(y)+ f_{S_t}(y) \circledast f_{U_{t}}(x - y)\big)\\
&\cdot
f_{U_c}(\tau+y)dy
d\tau f_X(x)dx,
\end{align*}
and $\xi(x) = \int_{x}^{\infty} f_{U_t}(u)du$.
\end{theorem}
\begin{proof}
The proof is given in Appendix~\ref{App: proof of average AoC soft extension}.
\end{proof}
\begin{remark}
Let $X_{k}$, $S_{k, t}$, and $S_{k, c}$ be exponentially distributed with parameters $\lambda$, $\mu_t$, and $\mu_c$, respectively. The expression in \eqref{eq: average AoC soft extension} reduces to that in \eqref{eq: AoC soft-1} and \eqref{eq: AoC soft-2}.
\end{remark}
\begin{remark}
Assumption~\ref{ass: distributions} provides the minimum sufficient conditions for the closed form of the average AoC. However, it does not guarantee the convergence of $\Theta_{\text{soft}}$, which heavily depends on the distributions of $X_k$, $S_{k, t}$, $S_{k, c}$, $U_{k, t}$, and $U_{k, c}$.
\end{remark}

\section{AoC under the Hard Deadline}\label{sec: AoC under Hard Deadlines}

In this section, we first provide graphical insights into the curve of the AoC (see Fig.~\ref{AoC-hard}). Next, we derive a general expression for the time-average AoC, presented in Theorem~\ref{thm: average AoC hard general} (Section~\ref{subsec: general expression hard}). Following this, we accurately approximate the expression in a fundamental case: the M/M/1-M/M/1 system, as detailed in Theorem~\ref{thm: average AoC hard} (Section~\ref{subsec: Approximated Average AoC in M/M/1 - M/M/1 Tandems}). In addition, we define computation throughput (see Definition~\ref{def: computation throughput} in Section~\ref{subsec: Computation Throughput}) and derive its expression (see Lemma~\ref{lem: computation throughput} in Section~\ref{subsec: Computation Throughput}). Subsequently, we investigate the trade-off between computation freshness and computation throughput (see Lemma~\ref{lem: pareto optimal} in Section~\ref{subsec: Computation Throughput}).
Finally, we generalize the accurate approximations for both the average AoC and computation throughput to broader cases, specifically G/G/1-G/G/1 systems, as presented in Theorem~\ref{thm: average AoC hard extension} and Proposition~\ref{pro: approximated CT extension} (Section~\ref{subsec: Extension hard}).

\subsection{General Expression for $\Theta_{\text{hard}}$}\label{subsec: general expression hard}

From \eqref{eq: AoC-hard} in Definition~\ref{def: AoC}, under the hard deadline, $c_{\text{hard}}(t)$ is {\it solely} determined by informative tasks. 
When $w=0$, all tasks are considered invalid, leading to $c_{\text{hard}}(t)=t$, which increases linearly with time $t$. Conversely, when $w=\infty$, there is no deadline, and all tasks are considered valid. In this case, $G(t)=N(t)$ for all $t$, so $c_{\text{hard}}(t)=t-\tau_{G(t)}$, which is only affected by the latest task.

\begin{figure}[htbp]
	\centering
	\includegraphics[width=0.42\textwidth,height=0.39\textwidth]{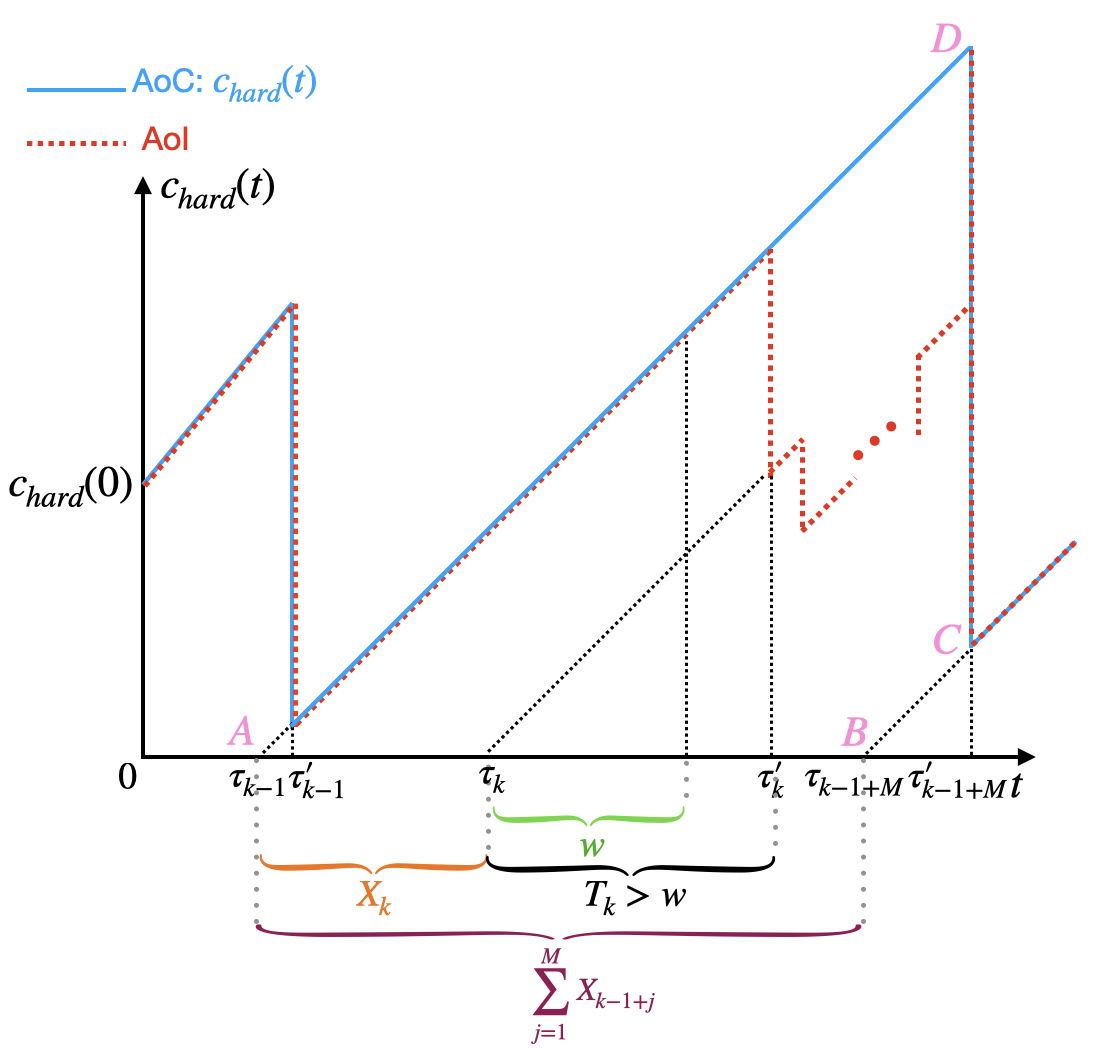}
	\caption{The curve of AoC under the hard deadline.}
	\label{AoC-hard}
\end{figure}

The AoC curve is depicted in Fig.~\ref{AoC-hard}, with the curve of AoI (see \cite{AoItutorial}) as a benchmark. In this figure, task $k-1$ is valid, so both AoI and AoC decreases at time $\tau_{k-1}'$. Suppose that after task $k-1$, the next valid task has the index $k-1+M$. Here, $M$ is a random variable with the distribution
\begin{align}\label{eq: Mk-1}
&\Pr\{M=n\}\nonumber\\
&=\Pr\{T_k>w,\cdots, T_{k-1+n-1}>w, T_{k-1+n}\leq w \}.
\end{align}
In \eqref{eq: Mk-1}, $M\geq 1$. Since the network is stationary, the random variable $M$ associated with every valid task has the identical distribution as in \eqref{eq: Mk-1}.
The AoC does not decrease at times $\tau_{k}'$, $\tau_{k+1}'$, $\cdots$, $\tau_{k-1+M-1}'$, and decreases at time $\tau_{k-1+M}'$. In the interval $[\tau_{k-1}', \tau_{k-1+M}')$, the AoC increases linearly with time $t$.

\begin{theorem}\label{thm: average AoC hard general}
The average AoC can be calculated as
\begin{align}\label{eq: average AoC hard}
	\Theta_{\text{hard}} = \frac{\mathbb{E}[T_M\cdot\sum_{j=1}^{M}X_j ]+\frac{1}{2}\mathbb{E}[\big(\sum_{j=1}^{M}X_j\big)^2]}{\mathbb{E}[\sum_{j=1}^{M}X_j]}.
\end{align}
\end{theorem}
\begin{proof}
The proof is given in Appendix~\ref{App: average AoC hard}.
\end{proof}
Although the general expression for the average AoC is given by \eqref{eq: average AoC hard}, a further exploration on this expression is challenging due to a couple of correlations involved. These correlations are detailed as follows.
\begin{itemize}
\item [(i)] Correlation between delays: In the queuing system, if the delay of task $k$, $T_k$, increases, the waiting time for task $k+1$ also increases, leading to a larger delay for task $k+1$, $T_{k+1}$. Hence, the sequence $\{T_k\}_k$ consists of identically distributed but positively correlated delays.
\item [(ii)] Correlation between delays and $M$ \big(defined in \eqref{eq: Mk-1}\big): According to \eqref{eq: Mk-1}, $M$ is the first index $n$ such that $T_n\leq w$ while all previous $T_k>w$. A higher value of $T_k$  suggests a higher likelihood of subsequent $T_j$ values (for $j>k$) also being high, thus making $M$ larger because it takes longer for a $T_k$ to
be less than or equal to $w$. This indicates that $T_k$ are $M$ are positively correlated. 
\item [(iii)] Correlation between the inter-arrival times $X_k$ and delays $T_k$:
If $X_k$ is larger, meaning that the inter-arrival time between task $k-1$ and task $k$ is longer, then $T_k$ is likely smaller because the waiting time for task $k$ is reduced. Therefore, $T_k$ and $X_k$ are negatively correlated. 
\item [(iv)] Correlation between the inter-arrival times $X_k$ and $M$: Since $X_k$ are $T_k$ are negatively correlated, and $T_k$ and $M$ are positively correlated, $X_k$ and $M$ are negatively correlated.
\end{itemize}

\subsection{Average AoC in M/M/1-M/M/1 Systems}\label{subsec: Approximated Average AoC in M/M/1 - M/M/1 Tandems}

Due to all these correlations, it is extremely challenging to derive the closed-form expression for $\Theta_{\text{hard}}$ in \eqref{eq: average AoC hard}. However, we can approximate it accurately under specific conditions.

\begin{theorem}\label{thm: average AoC hard}
When $\mu_t\gg\lambda$ and $\mu_c\gg\lambda$, the average AoC defined in \eqref{eq: average AoC hard} can be accurately approximated by $\hat{\Theta}_{\text{hard}}$, 
\begin{align}\label{eq: AoC hard-1}
\hat{\Theta}_{\text{hard}} =\mathbb{E}[T_M]+\frac{\mathbb{E}[X_1^2]}{2\mathbb{E}[X_1]}+\big(\frac{\mathbb{E}[M^2]}{2\mathbb{E}[M]}-\frac{1}{2}\big)\mathbb{E}[X_1].
\end{align}
Let $\rho_t$, $\rho_c$, $\delta_t$, $\delta_c$,
$\zeta_t$, and $\zeta_c$ be given in Theorem~\ref{thm: average AoC soft}, if $\mu_t\neq \mu_c$,
\begin{align}\label{eq: Thetahard-1}
\hat{\Theta}_{\text{hard}} =& \frac{\frac{1-\zeta_t(1+\mu_t\delta_t w)}{\mu_t^2\delta_t^2} - \frac{1-\zeta_c(1+\mu_c\delta_c w)}{\mu_c^2\delta_c^2}}{(1-\zeta_t)/\mu_t\delta_t-(1-\zeta_c)/\mu_c\delta_c}\nonumber\\
+&\frac{\mu_c-\mu_t}{\lambda\big(\mu_c\delta_c(1-\zeta_t) - \mu_t\delta_t(1-\zeta_c)\big)};
\end{align}
if $\mu_t=\mu_c$,
\begin{align}\label{eq: Thetahard-2}
\hat{\Theta}_{\text{hard}} =& \frac{\frac{2}{\mu_t\delta_t}-(\frac{2}{\mu_t\delta_t}+2w+\mu_t\delta_tw^2)\zeta_t }{1-\zeta_t(1+\mu_t\delta_tw)}\nonumber\\
+&\frac{1}{\lambda\big(1 - \zeta_t(1+\mu_t\delta_tw)\big)}.
\end{align}
\end{theorem}
\begin{proof}
The proof is given in Appendix~\ref{App: average AoC hard closedform}.
\end{proof}
\begin{remark}\label{remark: lower bound}
(Lower Bound) $\hat{\Theta}_{\text{hard}}$ in \eqref{eq: AoC hard-1} captures an extreme case where the positive correlations among $\{T_k\}_k$ are removed. Therefore, $\hat{\Theta}_{\text{hard}}$ in \eqref{eq: AoC hard-1} serves as a lower bound for $\Theta_{\text{hard}}$. This lower bound is approximately tight when $\mu_t\gg\lambda$ and $\mu_c\gg\lambda$.
\end{remark}
By exchanging $\mu_t$ and $\mu_c$ in \eqref{eq: Thetahard-1} and \eqref{eq: Thetahard-2}, we observe that $\Theta_{\text{soft}}$ remains unchanged. This indicates that $\Theta_{\text{soft}}$ is symmetric with respect to $(\mu_t, \mu_c)$.
From a mathematical standpoint, this symmetry implies that both communication latency and computation latency equally affect $\Theta_{\text{soft}}$. Therefore, in practical terms, to improve the computation freshness $\Theta_{\text{soft}}$, one can reduce either the communication latency or the computation latency, as both have the same impact on the overall freshness.

\subsection{Computation Throughput}\label{subsec: Computation Throughput}
Under the hard deadline, the frequency of informative tasks is influenced by two facts: the arrival rate $\lambda$ and the deadline $w$. We define the frequency of informative tasks as {\it computation throughput}. Formally, we have the following definition.
\begin{definition}\label{def: computation throughput}
(Computation Throughput) Let $K(t)$  denote the number of informative tasks within the interval $[0, t]$. The computation throughput is then defined as
\begin{align}\label{eq: computation throughput}
\Xi = \lim_{t\to\infty}\frac{K(t)}{t}.
\end{align}
\end{definition}

\begin{lemma}\label{lem: computation throughput}
The computation throughput is given by,
\begin{align}\label{eq: computation throughput-1}
\Xi = \frac{1}{\mathbb{E}[\sum_{k=1}^{M}X_k]}.
\end{align}
\end{lemma}
\begin{proof}
The proof is given in Appendix~\ref{App: computation throughput}.
\end{proof}
In \eqref{eq: computation throughput-1}, the arrival rate $\lambda$ is reflected in $X_k$, while the deadline $w$ is captured by $M$. It is worth noting that Definition~\ref{def: computation throughput} can apply to the case with a soft deadline. Under the soft deadline, all tasks are valid, so \eqref{eq: computation throughput} implies that $\Xi = \lim_{t \to \infty} \frac{K(t)}{t} = \lim_{t \to \infty} \frac{G(t)}{t} = \lambda$, which is a trivial case. Therefore, we did not investigate the computation throughput concept in Section~\ref{sec: AoC under Soft Deadlines}.

\begin{proposition}\label{pro: approximated CT}
When $\mu_t\gg\lambda$ and $\mu_c\gg\lambda$, the computation throughput defined in \eqref{eq: computation throughput-1} can be accurately approximated by $\hat{\Xi}$, 
\begin{align}\label{eq: TC}
\hat{\Xi} =\frac{1}{\mathbb{E}[M]\mathbb{E}[X_1]}.
\end{align}
Let $\rho_t$, $\rho_c$, $\delta_t$, $\delta_c$,
$\zeta_t$, and $\zeta_c$ be given in Theorem~\ref{thm: average AoC soft}, if $\mu_t\neq \mu_c$,
\begin{align}\label{eq: TC-1}
\hat{\Xi}=\lambda\cdot\frac{\mu_c\delta_c (1-\zeta_t)-\mu_t\delta_t(1-\zeta_c)}{\mu_c-\mu_t}, 
\end{align}
if $\mu_t=\mu_c$,
\begin{align}\label{eq: TC-2}
\hat{\Xi} =& \lambda(1-(1+\mu_t\delta_tw)\zeta_t).
\end{align}
\end{proposition}
\begin{proof}
The proof  is given in Appendix~\ref{App: approximated CT}.
\end{proof}
\begin{remark}\label{remark: upper bound}
(Upper Bound) 
According to \eqref{eq: Mk-1}, a higher value of $T_k$  suggests a higher likelihood of subsequent $T_j$ values (for $j>k$) also being high, thus making $M$ larger. However,
$\hat{\Xi}$ in \eqref{eq: TC} captures an extreme case where positive correlations among $\{T_k\}_k$ are removed, resulting in a smaller expectation for $M$. Consequently, $\hat{\Xi}$ in \eqref{eq: TC} serves an upper bound for $\hat{\Xi}$. Additionally, this upper bound is approximately tight when $\mu_t\gg\lambda$ and $\mu_c\gg\lambda$.
\end{remark}

A Pareto-optimal point represents a state of resources allocation where improving one objective necessitates compromising the other. A pair $(\hat{\Theta}^*, \hat{\Xi}^*)$ is defined as a Pareto-optimal point if, for any $(\hat{\Theta}, \hat{\Xi})$, both conditions (i) $\hat{\Theta}<\hat{\Theta}^*$ (or $\hat{\Theta}\leq \hat{\Theta}^*$) and (ii) $\hat{\Xi}\geq\hat{\Xi}^*$ (or $\hat{\Xi}>\hat{\Xi}^*$) cannot hold simulatenously \cite{JLXYSK2020}.  This indicates that reducing computation freshness is impossible without degrading computation throughput, and increasing computation throughput cannot occur without compromising computation freshness.  

While closed-form expressions for computation freshness \big(see \eqref{eq: average AoC hard}\big) and computation throughput \big(see \eqref{eq: computation throughput-1}\big) are unavailable, their relationship can be approximated using \eqref{eq: AoC hard-1} and \eqref{eq: TC}, the analysis focuses on weakly Pareto-optimal points rather than strict Pareto-optimal points \cite{Pareto}. Consider the following optimization problem:
\begin{align}\label{eq: pareto optimization}
&\min_{\lambda:\,\, \hat{\Xi} > u}\,\,\hat{\Theta}.
\end{align}
Let the corresponding $\hat{\Theta}$ and $\hat{\Xi}$ as $\hat{\Theta}(u)$ and $\hat{\Xi}(u)$, respectively. The tradeoff between  computation freshness and the computation throughput is explored in the following lemma.
\begin{lemma}\label{lem: pareto optimal}
The objective pair $\big(\hat{\Theta}(u), \hat{\Xi}(u)\big)$ is a weakly Pareto-optimal point.
\end{lemma}
\begin{proof}
The proof  is given in Appendix~\ref{App: lem: pareto optimal}.
\end{proof}

\subsection{Extension: Average AoC in G/G/1-G/G/1 Systems}\label{subsec: Extension hard}

In this section, we extend the average AoC in M/M/1-M/M/1 tandems  (see Section~\ref{subsec: Approximated Average AoC in M/M/1 - M/M/1 Tandems}) to general case, i.e., G/G/1 - G/G/1 tandems: Computational tasks arrive at the source via a random process characterized by an average rate of $\lambda$. Upon arrival, the transmission delay of each task follows a general distribution with an average rate of $\mu_t$, and the computation delay at the computational node follows a general distribution with an average rate of $\mu_c$.

\begin{theorem}\label{thm: average AoC hard extension}
When Assumption~\ref{ass: distributions} is satisfied, let $\mu_t$ and $\mu_c$ be in Theorem~\ref{thm: average AoC soft}. When $\mu_t\gg\lambda$ and $\mu_c\gg\lambda$, the average AoC defined in \eqref{eq: average AoC hard} can be accurately approximated by \eqref{eq: AoC hard-1}. In particular,
\begin{align}\label{eq: AoC hard-extension}
\hat{\Theta}_{\text{hard}} =& \frac{\int_{0}^{w}\tau\eta_1(\tau) d\tau}{F_T(w)}+\frac{\int_{0}^{\infty}x^2f_X(x)dx}{2\int_{0}^{\infty}xf_X(x)dx}\nonumber\\
+&\frac{1-F_T(w)}{F_T(w)}\cdot\int_{0}^{\infty}xf_X(x)dx.
\end{align}
where 
$F_T(w) = \int_{0}^{w}\eta_1(\tau) d\tau$ and $\eta_1(\tau)$ is defined in Theorem~\ref{thm: average AoC soft extension}.
\end{theorem}
\begin{proof}
The proof is given in Appendix~\ref{App: proof of average AoC hard extension}.
\end{proof}

\begin{proposition}\label{pro: approximated CT extension}
When Assuption~\ref{ass: distributions} is satisfied, let $\mu_t$ and $\mu_c$ be in Theorem~\ref{thm: average AoC soft}. When $\mu_t\gg\lambda$ and $\mu_c\gg\lambda$, the computation throughput defined in \eqref{eq: computation throughput-1} can be accurately approximated by \eqref{eq: TC}. In particular,
\begin{align}\label{eq: TC extension}
\hat{\Xi} =\frac{F_T(w)}{\int_{0}^{\infty}xf_X(x)dx},
\end{align}
where $F_T(w)$ is given in Theorem~\ref{thm: average AoC soft extension}.
\end{proposition}
\begin{proof}
The proof  is given in Appendix~\ref{App: approximated CT extension}.
\end{proof}

\section{AoC-based Scheduling for Multi-Source Networks}\label{sec: AoC-based Scheduling for Multi-Source 3CNs}

The AoC concept is not only utilized in continuous-time settings, but also suitable for discrete-time setting. In this section, we explore the application of the AoC concept to resource optimization and {\it real-time} scheduling in multi-source networks, providing an analysis of the resulting optimal scheduling policies. 

\subsection{System Model in Multi-Source Networks}\label{subsec: System Model in Multi-Source 3CNs}

Consider a network with a computational node processing tasks from $N$ sources, as illustrated in Fig.~\ref{Multisources}. Time is slotted, indexed by $k\in\{1,2,\cdots, T\}$, where $T$ is the time-horizon of this discrete-time system. In every time slot, computational tasks arrive source~$i$ with rate $\lambda_i$. Upon arrival, tasks are immediately available at the corresponding transmitter~$i$, where they enter a communication queue awaiting transmission. Transmitter~$i$ transmits tasks at a rate $\mu_{t, i}$: if a task is under transmission during a slot, the transmission completes with probability $\mu_{t, i}$ by the end of the slot. Once transmitted, the tasks are received by the receiver, where they await processing.  The computational node processes tasks at a rate $\mu_c$. Without loss of generality, we assume $\mu_c=1$, the framework can be straightforwardly extended to cases where $\mu_c<1$. After processing, tasks depart from the system. In most realistic scenarios, recent computational tasks hold greater importance than earlier ones, as they often contain fresher and more relevant information. To address this, we adopt a {\it preemptive} rule \cite{cxrAoI, cxrestimation}: newly arrived tasks can replace previously queued tasks already present in the transmitter.

\begin{figure}[htbp]
\centering
\includegraphics[width=0.39\textwidth,height=0.28\textwidth]{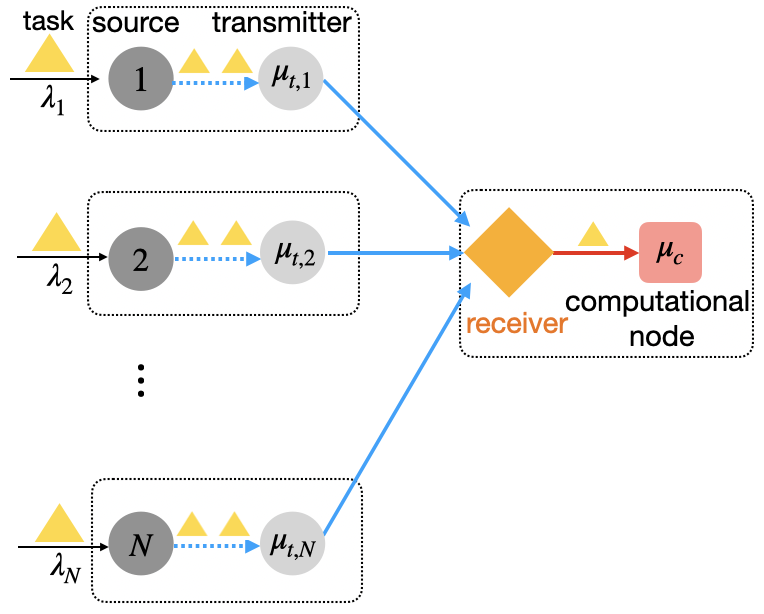}
\caption{A 3CN with multi sources/transmitters, a receiver, and a computational node.}
\label{Multisources}
\end{figure}

At each slot, the computational node either idles or selects a transmitter to transmits its task. Let $a_i(k)\in\{0, 1\}$ be an indicator function where  $a_i(k)=1$ if transmitter $i$ is selected for transmission in time $k$, and $a_i(k)=0$ otherwise. Similarly, let  $d_i(k)=1$ indicate that a task from transmitter~$i$ successfully leaves the computational node at slot $k$. Due to limited communication resources, {\it at most one} task can be transmitted across all transmitters in a given slot. Therefore, the following constraints hold: 
\begin{align}\label{eq: a and d}
\sum_{i=1}^{N} a_i(k)\leq 1,\quad \sum_{i=1}^{N} d_i(k)\leq 1,\quad\forall k.
\end{align}
Additionally, if a task successfully leaves the computational node, \big($\sum_{i=1}^{N}d_i(k)=1$\big), indicating that the last scheduling decision has been responded to, the computational node will schedule a transmitter in the next time slot \big($\sum_{i=1}^{N}a_i(k+1)=1$\big).  Here, since $\mu_c=1$, $\sum_{i=1}^{N}d_i(k)=1$ further implies that the current task reaches the receiver at the beginning of time slot $k$. Conversely, if the last scheduling decision has not been responded to in the current slot, \big(i.e., $\sum_{i=1}^{N}d_i(k)=0$\big), the computational node will not schedule a transmitter in the next slot \big($\sum_{i=1}^{N}a_i(k+1)=0$\big).  This relationship is formalized as: $\sum_{i=1}^{N} a_i(k+1) = \sum_{i=1}^{N} d_i(k)$.

Let $c^{(i)}_x(k)$ with $x\in\{\text{soft}, \text{hard}\}$ be the AoC associated with source $i$ at the end of slot $k$. The time-average AoC associated with source $i$ is given by $\mathbb{E}[\sum_{k=1}^{T}c^{(i)}_x(k)]/T$. For capturing the freshness of computation of this network employing scheduling policy $\pi\in\Pi$, we define the average sum AoC in the limit as the time-horizon grows to infinity as
\begin{align}\label{eq: average sum AoC}
\Theta_x=\lim_{T\to\infty}\frac{1}{TN}\sum_{k=1}^{T}\sum_{i=1}^{N}\mathbb{E}[c^{(i)}_x(k)], \,\, x\in\{\text{soft}, \text{hard}\}.
\end{align}
The AoC-optimal the scheduling policy $\pi^*\in\Pi$ is the one that minimizes the average sum AoC:
\begin{align}\label{eq: optimal average sum AoC}
\Theta_{x}^* = \min_{\pi\in\Pi} \Theta_{x}.
\end{align}

At the end of this subsection, we introduce the concept of instantaneous delay for transmitter~$i$ at the end of time slot $k$, denoted by $z_i(k)$. If a task is in transmitter~$i$ at time slot~$k$, $z_i(k)$ is defined as the delay experienced by that task up to time~$k$. If transmitter~$i$ is idle at time slot~$k$, we define  $z_i(k)=c^{(i)}_x(k)$. Since the computational node is the decision-making component of the centralized network, it is reasonable to assume that it has access to the information from each transmitter in every time slot, i.e., $z_i(k)$.

\subsection{AoC-based Max-Weight Policies}\label{subsec: AoC-based Max-Weight Policies}
Finding the global optimal policy for the optimization problem \eqref{eq: optimal average sum AoC} is challenging due to the real-time nature of scheduling decisions. Drawing inspiration from \cite{IKEM2021}, we employ Lyapunov Optimization to develop AoC-based Max-Weight policy. This policy have been shown to be near-optimal \cite{IKEM2021, Lyapunov}. The Max-Weight policy is designed to minimize the expected drift of the Lyapunov function in each time slot, thereby striving to reduce the AoC across the network.

We use the following linear Lyapunov Function:
\begin{align}\label{eq: Lyapunov function}
L(k)\triangleq \frac{1}{N}\sum_{i=1}^{N}\beta_i c^{(i)}_{x}(t),\,\, x\in\{\text{soft}, \text{hard}\}.
\end{align}
where $\beta_i$ is a positive hyperparameter that allows the Max-Weight policy to be tuned for different network configurations and queuing disciplines. The Lyapunov Drift is defined as
\begin{align}\label{eq: Drift}
\Delta\big(\mathcal{S}(k)\big) := \mathbb{E}[L(k+1) - L(k) |\mathcal{S}(k)],
\end{align}
where $\mathcal{S}(k)$ represents the network state at the beginning of time slots $k$ and $k-1$\footnote{Here, we define $\mathcal{S}(k)$ to represent the network state over both the beginning of time slots $k$ and $k-1$, rather than only at beginning of time slot $k$ as in \cite{IKEM2021}. This adjustment is made because when a transmitter is scheduled, the corresponding task requires at least $2$ time slots to complete the computation.}, which is defined by: for $ x\in\{\text{soft}, \text{hard}\}$,
\begin{align*}
\mathcal{S}(k) = \Big\{\big\{c^{(i)}_x(\tau), z_i(\tau), \{d_i(\tau)\}\big\}_{\tau=k-1}^{k}\Big\}_{i=1}^N.
\end{align*}

The Lyapunov Function $L(k)$ increases with the AoC of the network,  while the Lyapunov Drift  $\Delta\big(\mathcal{S}(k)\big)$ represents the expected change in $L(k)$ over two time slots. By minimizing the drift \eqref{eq: Drift} in each time slot, the Max-Weight policy aims to maintain both $L(k)$ and the network’s AoC at low levels.

\subsubsection{Under the Soft Deadline} 
To simplify \eqref{eq: Lyapunov function}, we derive the AoC in time slot $k+1$ under the soft deadline assumption. Recall that $\mu_c=1$, meaning computational tasks are processed immediately upon arrival at the computational node without queuing. Let $\ell_i\big(z_i(k)\big)=\frac{A_i(k)}{G_i(k)}(z_i(k)+1-w)^+$, where $G_i(k)$ is defined in \eqref{eq: index of information packet} and $A_i(k)$ is defined in \eqref{eq: index of delay not exceeding w}. We can derive the expression for $c^{(i)}_{\text{soft}}(k+1)$ as follows (the proof is given by Appendix~\ref{App: csoftt+2}):
\begin{align}\label{eq: recursion c soft-1}
&c^{(i)}_{\text{soft}}(k+1) = 1_{\{\sum_{i=1}^{N}d_i(k-1)=0\}}\big(c^{(i)}_{\text{soft}}(k) + 1\big)\nonumber\\
&+1_{\{\sum_{i=1}^{N}d_i(k-1)=1, a_i(k)=0\}}\big(c^{(i)}_{\text{soft}}(k) + 1\big)\nonumber\\
&+1_{\{\sum_{i=1}^{N}d_i(k-1)=1, a_i(k)=1, d_i(k+1)=0\}}\big(c^{(i)}_{\text{soft}}(k) + 1\big)\nonumber\\
&+1_{\{\sum_{i=1}^{N}d_i(k-1)=1, a_i(k)=1, d_i(k+1)=1\}}\nonumber\\
&\cdot\Big(z_i(k) +1+ \ell_i\big(z_i(k)\big)\Big)
\end{align}
From \eqref{eq: Lyapunov function}, \eqref{eq: Drift}, and \eqref{eq: recursion c soft-1}, we propose the Max-Weight algorithm, i.e.,  Algorithm~\ref{alg: soft}, and prove that it minimizes the Lyapunov Drift in Proposition~\ref{pro: maxweight soft}. 

\begin{algorithm}
\caption{Max-Weight Policy for Soft Deadline}\label{alg: soft}
	\begin{algorithmic}[1]
		\Require $T$, $\{\beta_i\}_{i=1}^{N}$, $\{\mu_{t, i}\}_{i=1}^{N}$, $\{c^{(i)}_{\text{soft}}(0)\}_{i=1}^{N}$,
		\For{$1\leq k\leq T$}
		\If{$\sum_{i=1}^{N}d_i(k-1)=0$}
		\State $a_i(k)=0$ for $i\in\{1,2,\cdots,N\}$.
		\ElsIf{$\sum_{i=1}^{N}d_i(k-1)=1$}
		\State Calculate  $w_i(k)\triangleq c^{(i)}_{\text{soft}}(k)-z_i(k)-\ell_i\big(z_i(k)\big)$.
		\State Set $a_{i^*}(k)=1$ if $i^* = \arg\max_{i}\big(\beta_i\mu_{t, i}w_i(k)\big)$.
		\State If multiple $i^*$ exist, select one randomly.
		\EndIf
		\EndFor
		\State {\bf Output} $\Theta_{\text{soft}}$.
	\end{algorithmic}
\end{algorithm}

\begin{proposition}\label{pro: maxweight soft}
Algorithm~\ref{alg: soft} minimizes the Lyapunov Drift in every time slot.
\end{proposition}
\begin{proof}
The proof is given in Appendix~\ref{App: maxweight soft}.
\end{proof}

\subsubsection{Under the Hard Deadline}
To obtain the Max-Weight policies, we utilize the recursion of $c^{(i)}_{\text{hard}}(k)$ to minimize the Lyapunov Drift defined in \eqref{eq: Drift}.

By similar process in Appendix~\ref{App: csoftt+2}, we can derive recursion for $c^{(i)}_{\text{hard}}(k+1)$ is given as follows:
\begin{align}\label{eq: recursion c hard}
&c^{(i)}_{\text{hard}}(k+1) = 1_{\{\sum_{i=1}^{N}d_i(k-1)=0\}}\big(c^{(i)}_{\text{hard}}(k) + 1\big)\nonumber\\
&+1_{\{\sum_{i=1}^{N}d_i(k-1)=1, a_i(k)=0\}}\big(c^{(i)}_{\text{hard}}(k) + 1\big)\nonumber\\
&+1_{\{\sum_{i=1}^{N}d_i(k-1)=1, a_i(k)=1, d_i(k+1)=0\}}\big(c^{(i)}_{\text{hard}}(k) + 1\big)\nonumber\\
&+1_{\{\sum_{i=1}^{N}d_i(k-1)=1, a_i(k)=1, d_i(k+1)=1\}}\nonumber\\
&\cdot\Big(1_{\{z_i(k)+1> w\}}\big(c^{(i)}_{\text{hard}}(k)+ 1\big)+1_{\{z_i(k)+1\leq w\}}\big(z_i(k) + 1\big)\Big).
\end{align}
From \eqref{eq: Lyapunov function}, \eqref{eq: Drift}, and \eqref{eq: recursion c hard}, we propose the Max-Weight algorithm, i.e.,  Algorithm~\ref{alg: hard}, and prove that it minimizes the Lyapunov Drift in Proposition~\ref{pro: maxweight hard}. 
\begin{algorithm}
	\caption{Max-Weight Policy for Hard Deadline}\label{alg: hard}
	\begin{algorithmic}[1]
		\Require $T$, $\{\beta_i\}_{i=1}^{N}$, $\{\mu_{t, i}\}_{i=1}^{N}$, $\{c_i^{\text{hard}}(0)\}_{i=1}^{N}$,
		\For{$1\leq k\leq T$}
		\If{$\sum_{i=1}^{N}d_i(k-1)=0$}
		\State $a_i(k)=0$ for $i\in\{1,2,\cdots,N\}$.
		\ElsIf{$\sum_{i=1}^{N}d_i(k-1)=1$}
		\State Calculate $w_i(k)\triangleq 1_{\{z_i(k)+1\leq w\}} \cdot\big(c^{(i)}_{\text{hard}}(k)-z_i(k)\big)$.
		\State Set $a_{i^*}(k)=1$ if $i^* = \arg\max_{i}\big(\beta_i\mu_{t, i}w_i(k)\big)$.
		\State If multiple $i^*$ exist, select one randomly.
		\EndIf
		\EndFor
		\State {\bf Output} $\Theta_{\text{hard}}$.
	\end{algorithmic}
\end{algorithm}

\begin{proposition}\label{pro: maxweight hard}
	Algorithm~\ref{alg: hard} minimizes the Lyapunov Drift  in every time slot.
\end{proposition}
\begin{proof}
The proof is similar to Appendix~\ref{App: maxweight soft}.
\end{proof}

\section{Numerical Results}\label{sec: simulations}
In this section, we verify our findings in Section~\ref{sec: AoC under Soft Deadlines}, Section~\ref{sec: AoC under Hard Deadlines}, and  Section~\ref{sec: AoC-based Scheduling for Multi-Source 3CNs} through simulations. 

\subsection{Simulations for AoC under the Soft Deadline}\label{subsec: simulations in soft}

In Fig.~\ref{AoCs}, we compare theoretical and simulated average AoC under the soft deadline. The parameters are set as $w=0.5$, $\mu_c=3$, $\mu_t\in\{2, 3\}$, and $\lambda\in (0, 1.8]$. The results show that the theoretical average AoC values derived in \eqref{eq: AoC soft-1} and \eqref{eq: AoC soft-2} closely match the simulation outcomes. This alignment implies that the theoretical expressions in Theorem~\ref{thm: average AoC soft} are accurate for both cases where $\mu_t\neq\mu_c$ and $\mu_t=\mu_c$. Additionally, in both scenarios, the average AoC initially decreases and then increases. When  $\mu_t=2<3=\mu_c$, the computation capability exceeds the communication capability. If $\lambda\to0$, the communication queue is almost empty, meaning that the communication resources (and thus the computation resources) are underutilized, resulting in a high average AoC. Conversely, if $\lambda\to\mu_t$, the communication queue is busy, meaning that the communication capability is utilized almost to its total capacity, resulting in many tasks waiting in the communication queue to be transmitted, which also leads to a high average AoC. Since $\mu_t$ and $\mu_c$ are symmetric in \eqref{eq: AoC soft-1} and \eqref{eq: AoC soft-2}, this conclusion is valid for the case when $\mu_t\geq\mu_c$. Using numerical methods (e.g., gradient methods in as described in \cite{convex}), we can find the optimal $\lambda$ in \eqref{eq: AoC soft-1} and \eqref{eq: AoC soft-2}. 

\begin{figure}[htbp]
\centering
\includegraphics[width=.4\textwidth]{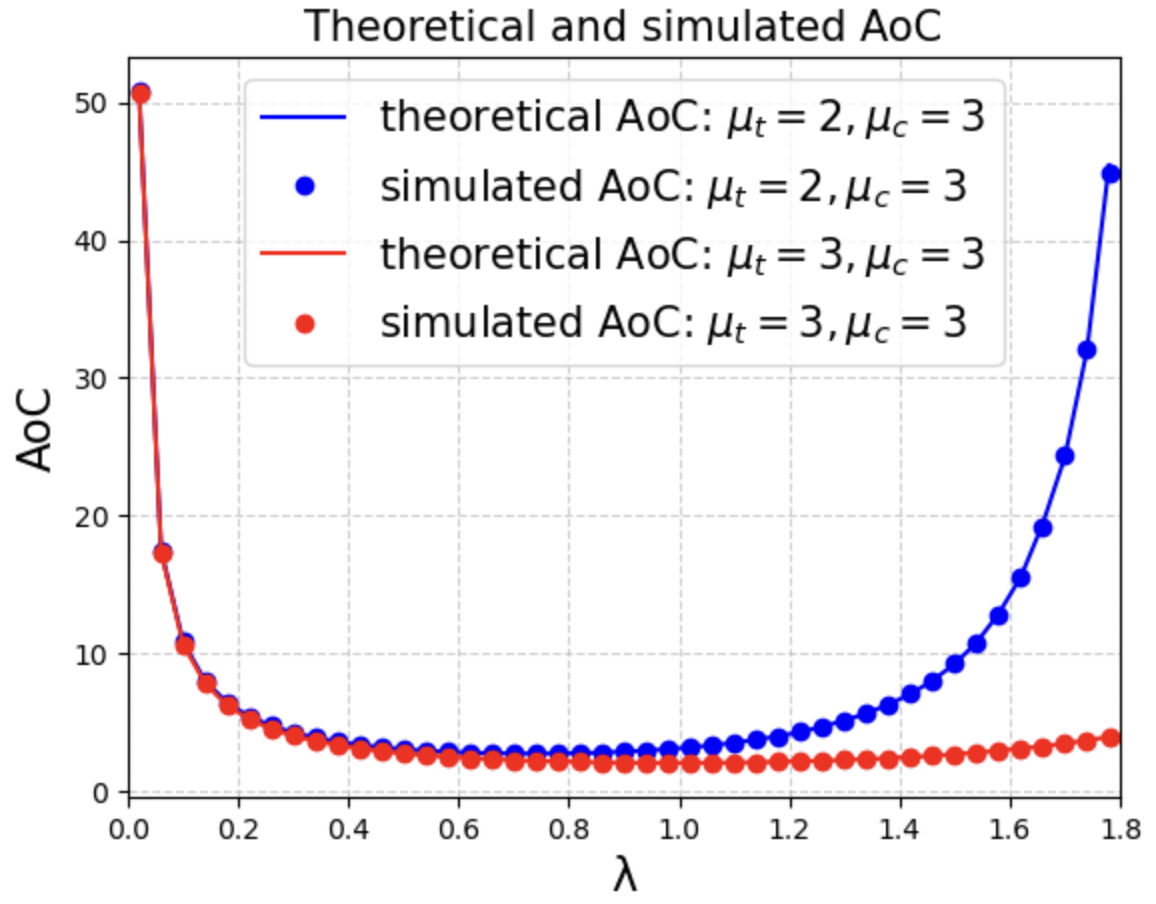}
\caption{Average AoC v.s. $\lambda$ under the soft deadline}
\label{AoCs}
\end{figure}

\begin{figure}[htbp]
\centering
\subfigure[Approximated and simulated AoC under different cases.] {\includegraphics[width=.4\textwidth]{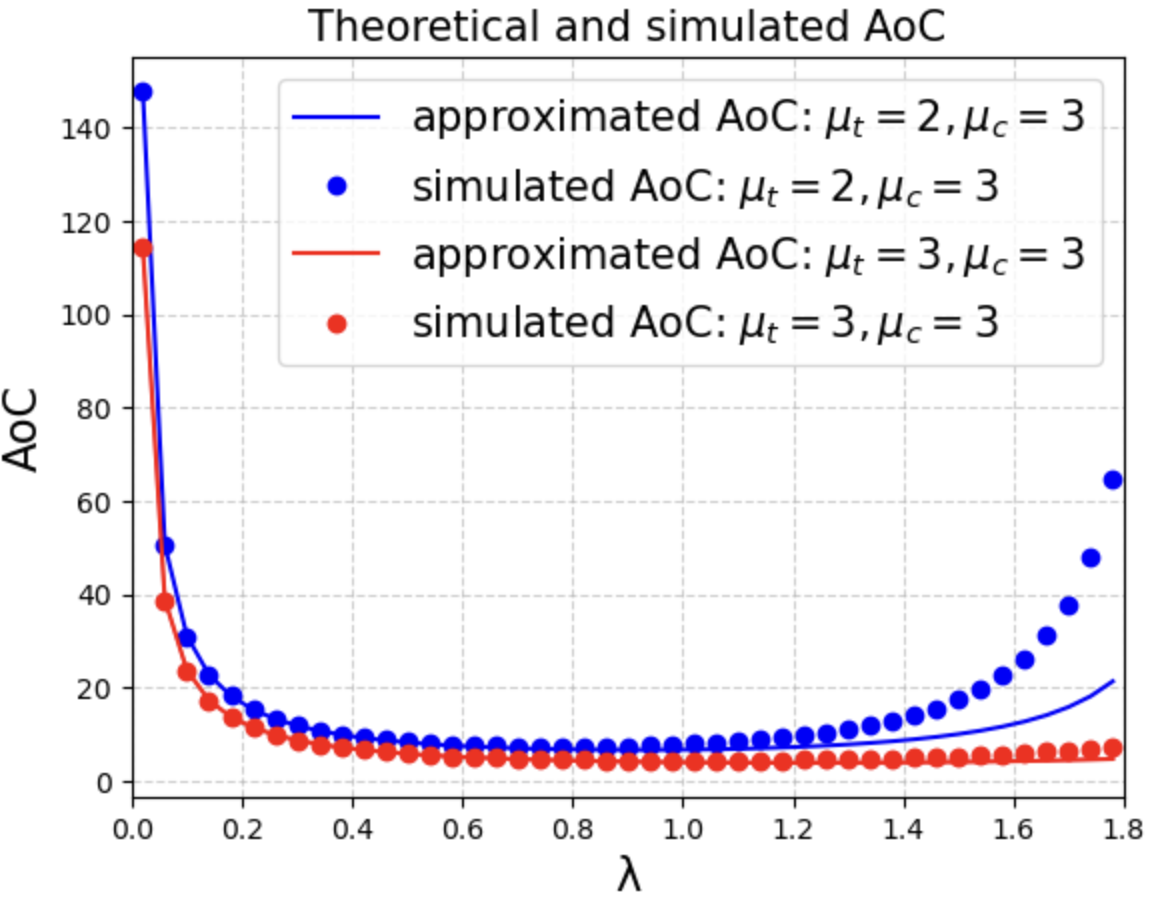}}
\subfigure[Approximated and simulated computation throughput when $\mu_t=2, \mu_c=3$.] {\includegraphics[width=.4\textwidth]{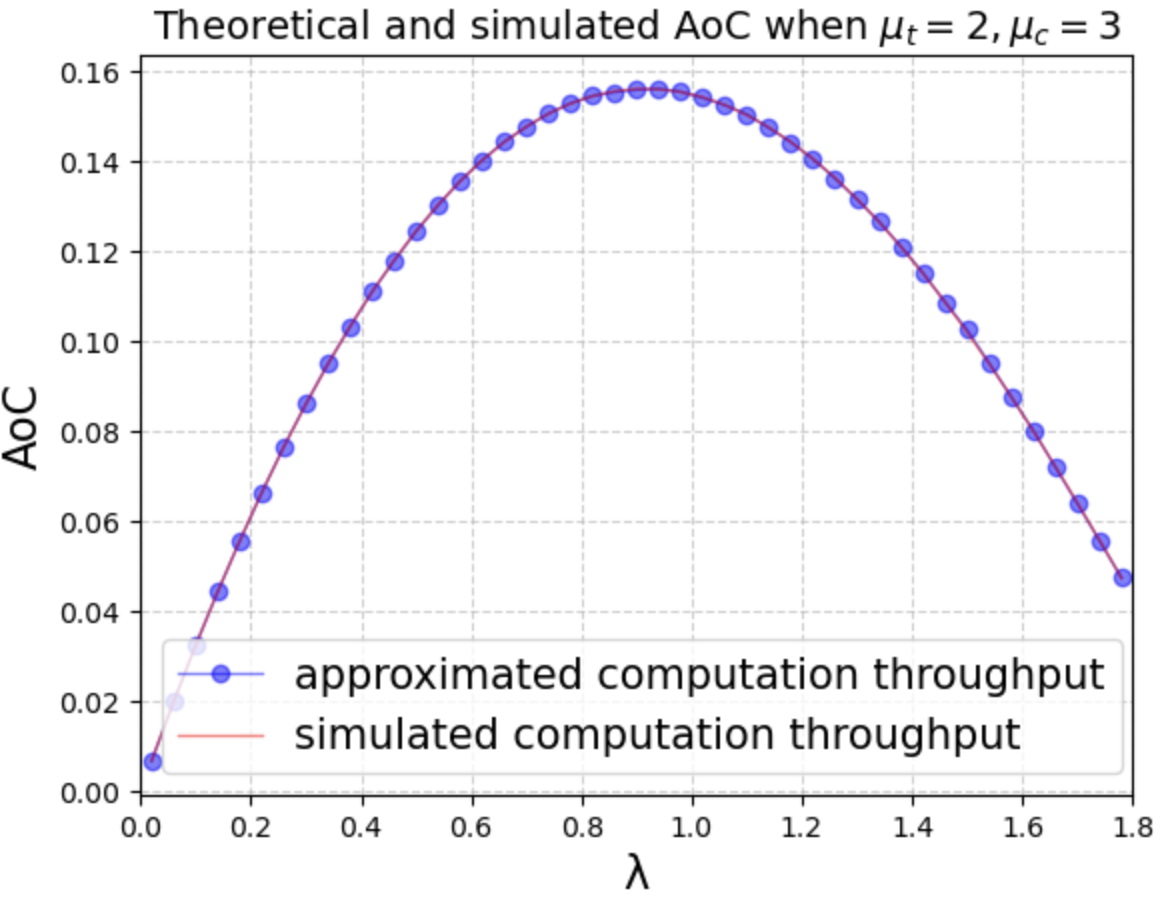}}
\caption{Numerical results for the average AoC under the hard deadline}
\label{AoCh}
\end{figure}

\subsection{Simulations for AoC under the Hard Deadline}\label{subsec: simulations in hard}

In Fig.~\ref{AoCh}, we investigate the average AoC under the hard deadline scenario. The comparisons between the approximated and simulated average AoC are presented in Fig.~\ref{AoCh}(a). We observe that when $\mu_t \gg \lambda$ and $\mu_c \gg \lambda$, the approximations of the average AoC in \eqref{eq: Thetahard-1} and \eqref{eq: Thetahard-2} closely match the simulation results. This agreement verifies the theoretical expressions in Theorem~\ref{thm: average AoC hard}.
Furthermore, the curves of the approximated average AoC are below those of the simulated average AoC, implying that the approximations in \eqref{eq: Thetahard-1} and \eqref{eq: Thetahard-2} serve as lower bounds for the average AoC, confirming the validity of Remark~\ref{remark: lower bound}. Additionally, in both scenarios, the average AoC initially decreases and then increases, similar to the trend observed under the soft deadline. This suggests that the computation freshness exhibits a convex relationship with respect to $\lambda$, indicating the existence of an optimal $\lambda$ that minimizes the average AoC.

We analyze the computation throughput as defined in \eqref{eq: computation throughput-1} in Fig.~\ref{AoCh}(b). The approximated computation throughput in \eqref{eq: TC} closely matches the simulated results, indicating the accuracy of the approximation. Although the positive correlations among $\{T_k\}_k$ are removed in \eqref{eq: TC}, making it a lower bound, the reduction in the expectation of $M$ is small. Consequently, the  reduction in $\mathbb{E}[\sum_{j=1}^{M}X_j]$ is negligible. Therefore, \eqref{eq: TC} serves as a tight lower bound for the computation throughput, not only when $\mu_t\gg\lambda$ and $\mu_c\gg\lambda$, but also across the entire range of $\lambda$.

\begin{figure}[htbp]
\centering
\subfigure[Relationship between the average AoC and  computation throughput.] {\includegraphics[width=.4\textwidth]{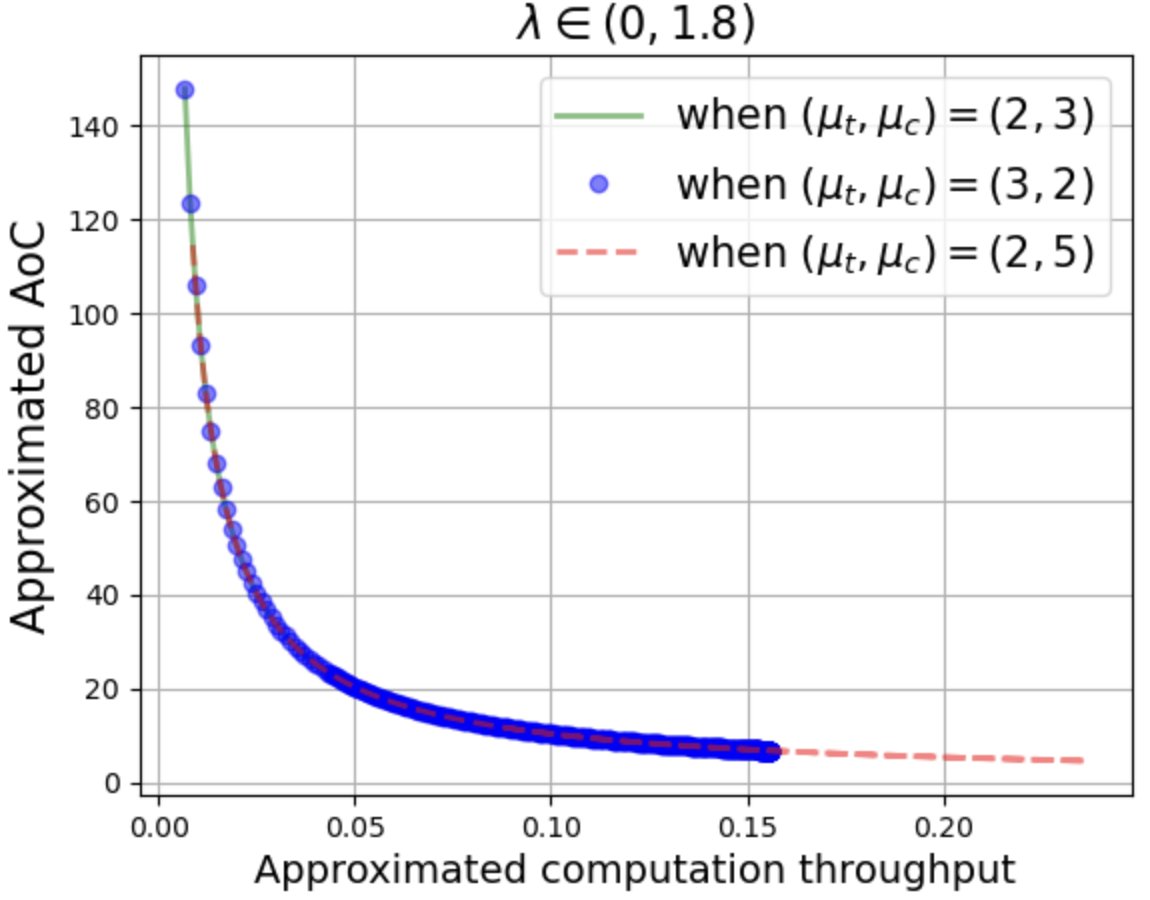}}
\subfigure[Relationship between the optimal average AoC and computation throughput.] {\includegraphics[width=.4\textwidth]{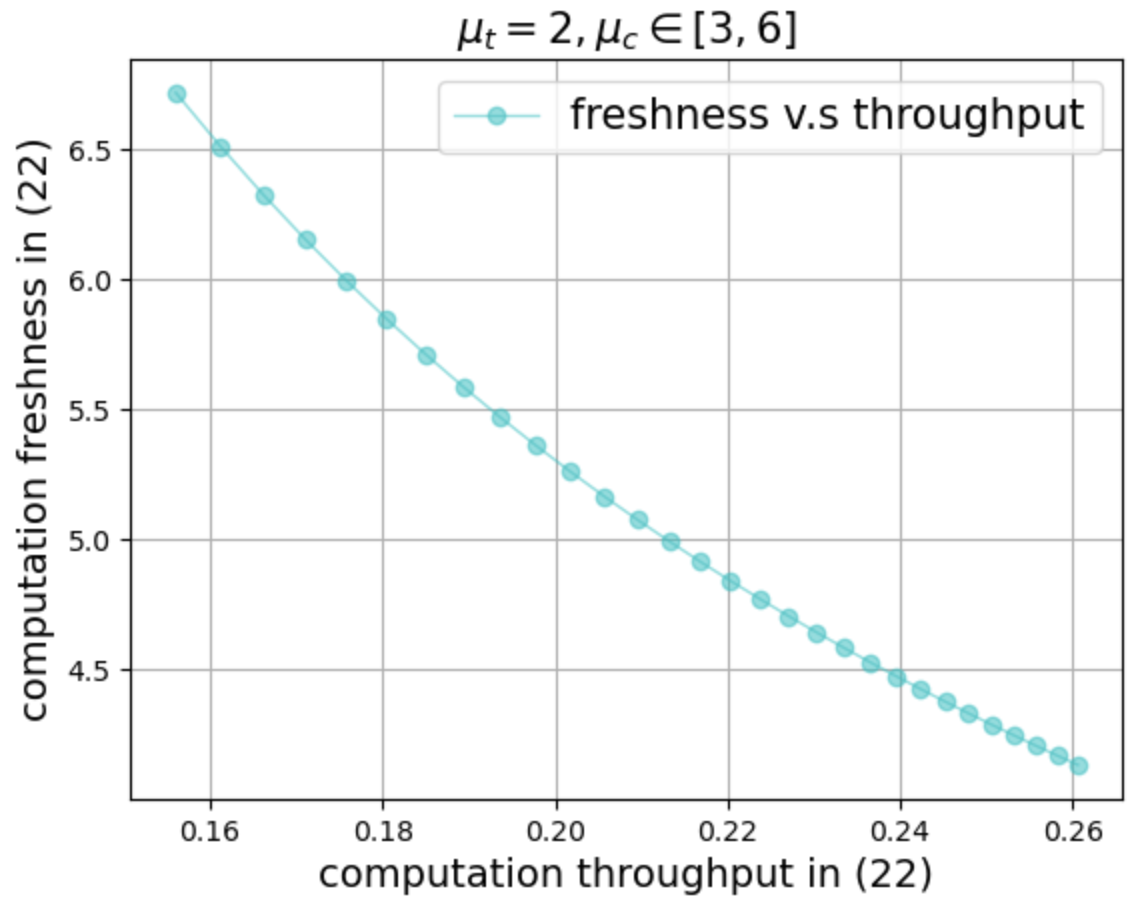}}
\caption{Relationships between computation freshness and computation throughput under the hard deadline.}
\label{tradeoff}
\end{figure}

The relationship between the computation freshness and the computation throughput is examined in  Fig.~\ref{tradeoff}. In Fig.~\ref{tradeoff}(a), we plot $\hat{\Theta}$ from \eqref{eq: AoC hard-1} versus $\hat{\Xi}$ from \eqref{eq: TC} when  $\lambda\in(0, 1.8)$ and $(\mu_t, \mu_c) = (2, 3), (3,2), (2, 5)$. These curves are observed to be decreasing, indicating that both $\hat{\Theta}$ and $\hat{\Xi}$ achieve their optimal values simultaneously.  Fig.~\ref{tradeoff}(a) shows that the approximated AoC decreases with the approximated computation throughput consistently. This suggests that computation throughput serves as a reliable proxy for the average AoC. Practically, we can minimize the average AoC by maximizing the computation throughput. 

Additionally, the comparison of $\hat{\Theta}$ versus $\hat{\Xi}$ for $(\mu_t, \mu_c) = (2, 3)$ (green solid line) aligns with that for $(\mu_t, \mu_c) = (3, 2)$ (blue scatter line). This demonstrates that interchanging the communication and computation rates does not alter the relationship, highlighting the symmetry with respect to $(\mu_t, \mu_c)$ as discussed below Theorem~\ref{thm: average AoC hard}.
Comparing $\hat{\Theta}$ v.s $\hat{\Xi}$ when $(\mu_t, \mu_c) = (2, 3)$ (the green solid line) with $\hat{\Theta}$ v.s $\hat{\Xi}$ when $(\mu_t, \mu_c) = (2, 5)$ (the red dashed line), we observe that, for any fixed $\lambda$, the approximated computation throughput under $(\mu_t, \mu_c) = (2, 5)$ is larger than that under $(\mu_t, \mu_c) = (2, 3)$. Similarly, the approximated average AoC under $(\mu_t, \mu_c) = (2, 5)$ is smaller than that under $(\mu_t, \mu_c) = (2, 3)$. This is as expected because larger communication or computation capabilities lead to better computation freshness and higher throughput. Both the optimal approximated average AoC and approximated computation throughput for $(\mu_t, \mu_c) = (2, 5)$ outperform those for $(\mu_t, \mu_c) = (2, 3)$ due to the consistency between AoC and computation throughput.

Finally, we solve the optimization \eqref{eq: pareto optimization} when $u=0$, and plot the optimal $\hat{\Theta}$ v.s the corresponding $\hat{\Xi}$ in Fig.~\ref{tradeoff}(b). We set $\mu_t=2$, $\lambda\in(0, 2)$, allowing $\mu_c$ to vary from $3$ to $6$. As expected, as the computation rate $\mu_c$ (representing the computation power of the network) increases, the optimal (approximated) average AoC decreases while the corresponding (approximated) computation throughput increases.

\subsection{AoC-based Scheduling in Multi-Source Networks}\label{subsec: AoC-based Scheduling in Multi-Source Networks}

In this subsection, we numerically validate the time-discrete and real-time Max-Weight scheduling policies proposed in Section~\ref{sec: AoC-based Scheduling for Multi-Source 3CNs}. To provide performance benchmarks, we consider two alternative policies: the Maximum Age First (MAF) policy~\cite{IKEM2021} and the stationary randomized policy~\cite{Yates}. Under the MAF policy, the receiver schedules the transmitter with the maximum age whenever a transmission opportunity arises. In contrast, under the stationary randomized policy, transmitter~$i$ is selected with a fixed probability $q_i$, and the receiver remains idle with probability $q_0$. We consider a symmetric network where $\lambda_i = \lambda$ for all $i$. The simulation parameters are set as follows:  number of sources $N=5$, computation service rate $\mu_c=1$, transmission service rate $\mu_{t, i}=0.5$ with $i\in\{1,2,\cdots,N\}$, deadline $w \in \{4, 10\}$, and arrival rate $\lambda \in [0.1, 0.9]$. Given the symmetry, we set $q_i = \frac{1}{N}$ and $\beta_i = 1$ for all $i$. 

\begin{figure}[htbp]
	\centering
	\subfigure[$w=4$.] {\includegraphics[width=.42\textwidth]{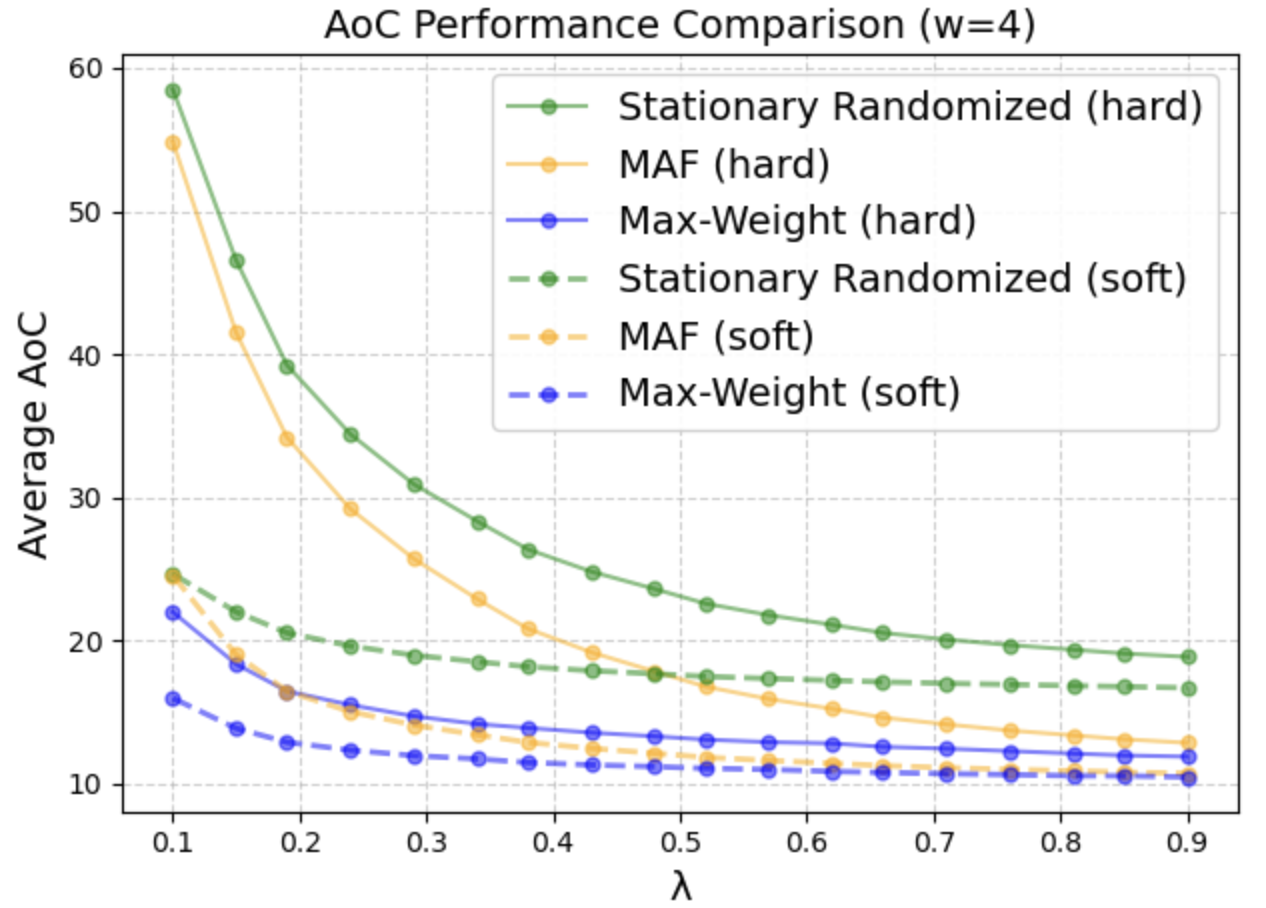}}
	\subfigure[$w=10$.] {\includegraphics[width=.42\textwidth]{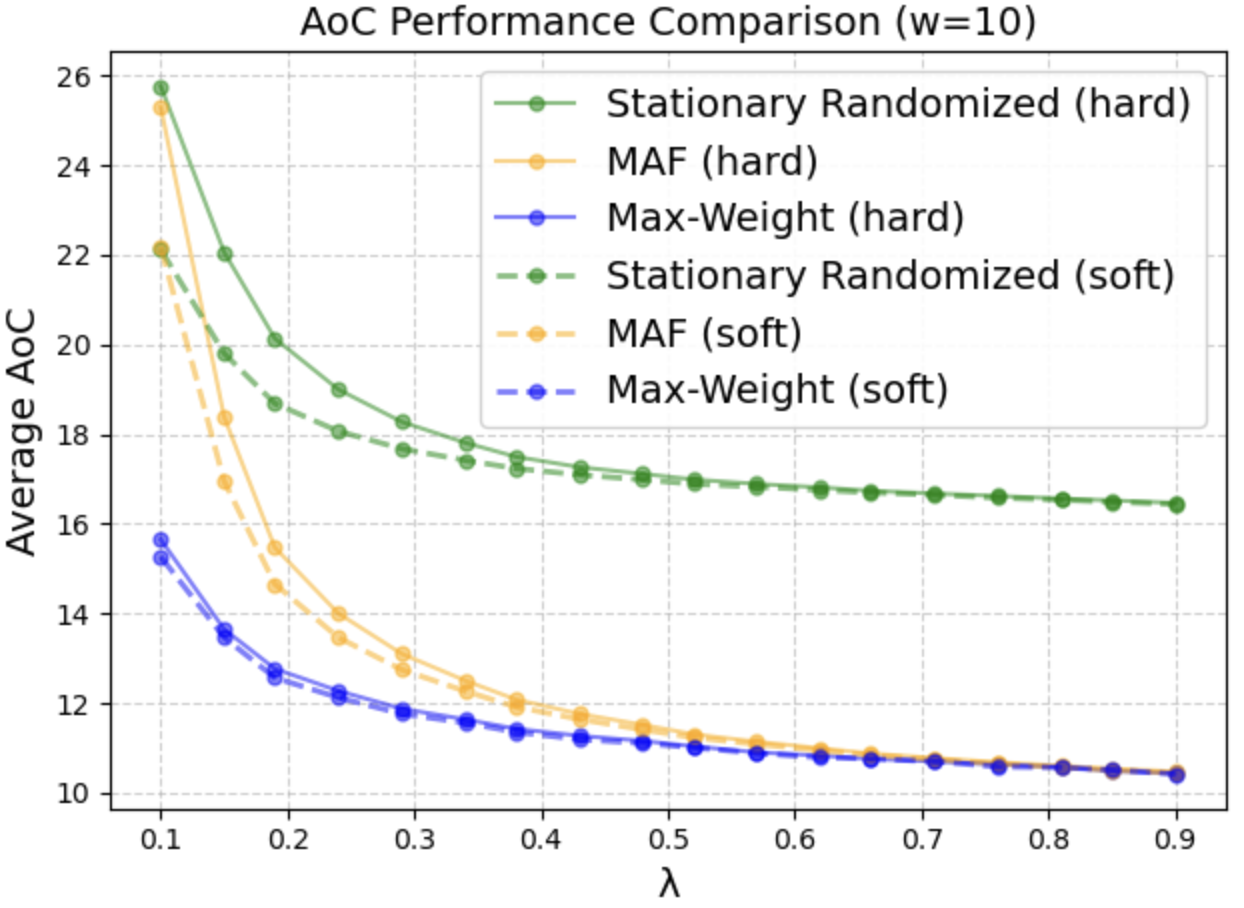}}
	\caption{Max-Weight policies and stationary randomized policies in multi-source networks.}
	\label{MW}
\end{figure}
	
As shown in Fig.~\ref{MW}, the average AoC decreases consistently as $\lambda$ increases. This behavior contrasts with that observed in Fig.~\ref{AoCs} and Fig.~\ref{AoCh}(a), where the average AoC initially decreases and then increases with increasing $\lambda$. This contrasting behavior stems from the difference in queuing disciplines: while Fig.~\ref{AoCs} and Fig.~\ref{AoCh}(a) adopt a first-come, first-served (FCFS) rule, Fig.~\ref{MW} incorporates a preemptive queuing rule. Under the preemptive rule, newly arrived tasks can replace previously queued tasks at the transmitter, thereby enhancing the freshness and reducing the AoC.

From Fig.~\ref{MW}(a) and Fig.~\ref{MW}(b), we observe that the Max-Weight policy consistently outperforms both benchmark policies.  Under the soft deadline setting, the performance gap between the Max-Weight and the MAF policies narrows as $w$ becomes larger, i.e.,  the difference between the dashed yellow and dashed blue curves in Fig.~\ref{MW}(a) is larger than that in Fig.~\ref{MW}(b). This trend can be explained as follows: When $w$ is large, the extended deadline allows more tasks to be considered valid even if they incur moderate delays.  Then, the additional latency  $1_{\{P(t)>G(t)\}}\cdot\frac{A(t)}{G(t)}\cdot(t - \tau_{P(t)}-w)^+$ \big(see \eqref{eq: AoC-soft}\big) tends to become small. Consequently,   the maximum weight in each time slot (as defined by the Max-Weight policy) is close to the maximum age used in the MAF policy. This alignment leads to similar scheduling decisions, thereby reducing the performance difference between the two policies.

Finally, comparing Fig.~\ref{MW}(a) and Fig.~\ref{MW}(b), we find that when $w$ is small (e.g., $w = 4$), the average AoC under the hard deadline is significantly higher than under the soft deadline. This is due to the stricter validity criteria of the hard deadline, which results in fewer tasks being considered valid, thus increasing the AoC a lot. However, when $w$ is large, most tasks meet the deadline requirement regardless of whether the hard or soft criterion is used, reducing the performance gap. As a result, the average AoC under hard and soft deadlines become increasingly similar as $w$ grows.

\section{Conclusion and Future Directions}\label{sec: conclusion}

In this paper, we introduce a novel metric AoC to quantify computation freshness in 3CNs. The AoC metric relies solely on tasks' arrival and completion timestamps, ensuring its applicability to dynamic and complex real-world 3CNs. We investigate AoC in two distinct setting. In point-to-point networks: tasks are processed sequentially with a first-come, first-served discipline. 
We derive closed-form expressions for the time-average AoC under  both a simplified case involving M/M/1-M/M/1 systems and a general case with G/G/1-G/G/1 systems. Additionally, we define the concept of computation throughput and derive its corresponding expressions.
We further apply the AoC metric to resource optimization and real-time scheduling in time-discrete multi-source networks. 
In the multi-source networks:  we propose AoC-based Max-Weight policies, leveraging a Lyapunov function to minimize its drift. Simulation results are provided to compare the proposed policies against benchmark approaches, demonstrating their effectiveness.

There are two primary future research directions: (1)
Deriving and analyzing time-average AoC in practical scenarios involving complex graph structures, such as sequential dependency graphs, parallel dependency graphs, and general dependency graphs \cite{8016573}.
(2) Exploring optimal AoC-based resource management policies in dynamic and complex 3CNs, such as task offloading in mobile edge networks.

\bibliographystyle{IEEEtran}
\bibliography{references}

\begin{IEEEbiography}[{\includegraphics[width=1in,height=1.25in,clip,keepaspectratio]{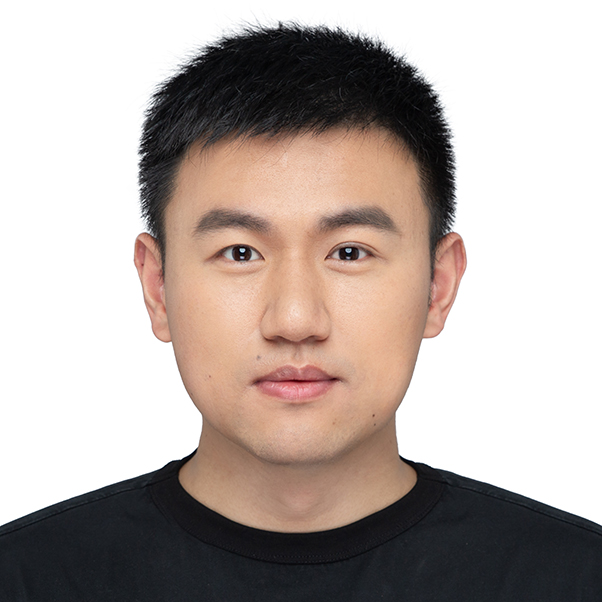}}]{Xingran Chen} received the M.A. in Applied Mathematics and Computational Science and the Ph.D. degree in Electrical and Systems Engineering from the University of Pennsylvania. He is currently a postdoctoral researcher at Rutgers University and an Assistant Professor (currently on leave) in the School of Information and Communication Engineering at University of Electronic Science and Technology of China. His research focuses on the theoretical foundations and collaborative learning algorithms for decentralized systems. He received the IEEE Communications Society \& Information Theory Society Joint Paper Award in 2023. He served as a guest editor for China Communications in 2024, and Entropy in 2025.
\end{IEEEbiography}

\begin{IEEEbiography}[{\includegraphics[width=1in,height=1.25in,clip,keepaspectratio]{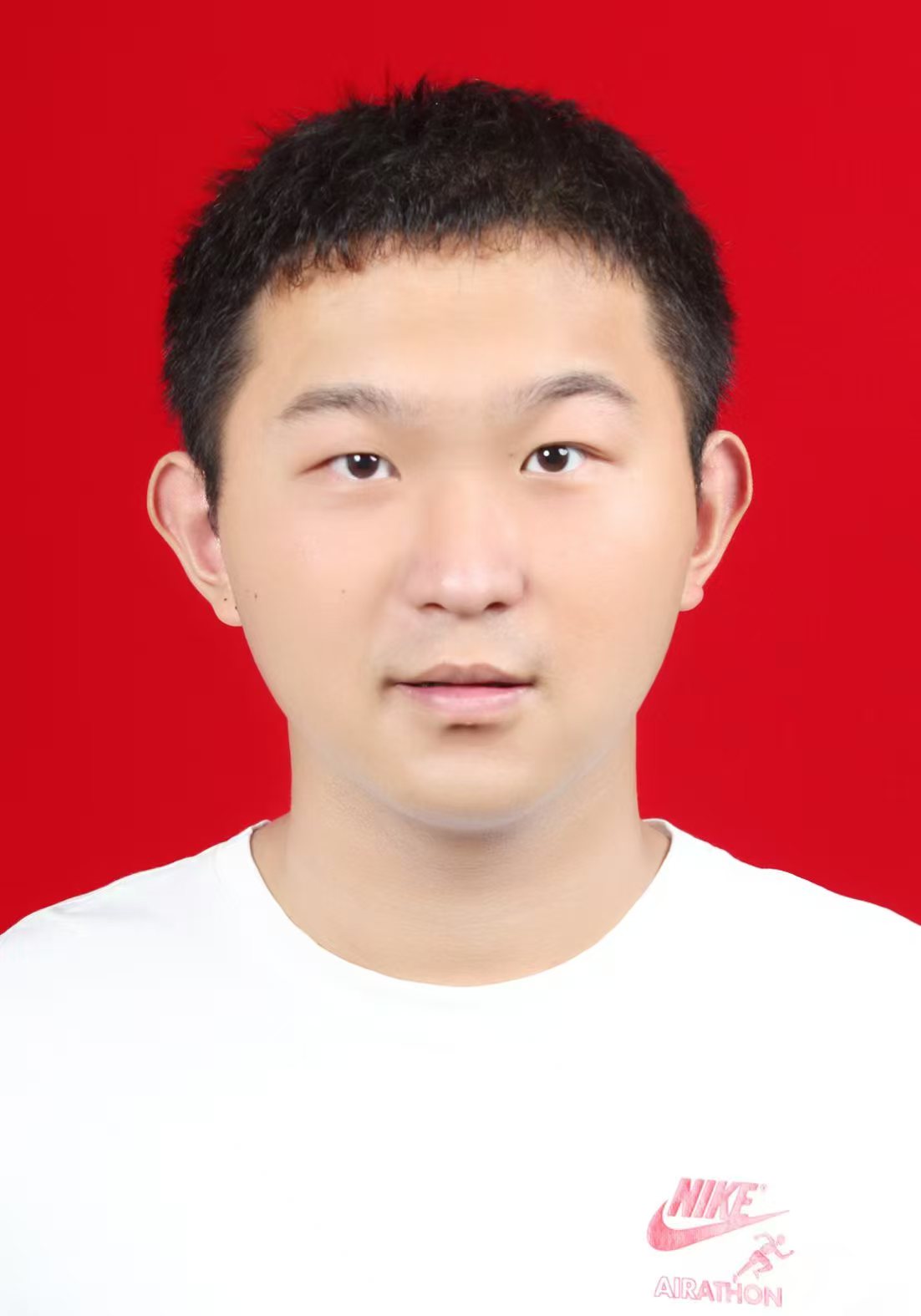}}]{Yi Zhuang} 
received the B.S. degree in communication engineering from Sichuan University, Chengdu, China, in 2024. He is currently pursuing the M.S. degree in information and communication engineering at the University of Electronic Science and Technology of China. His research interests include distributed network systems, edge and cloud computing, and reinforcement learning.
\end{IEEEbiography}

\begin{IEEEbiography}[{\includegraphics[width=1in,height=1.25in,clip,keepaspectratio]{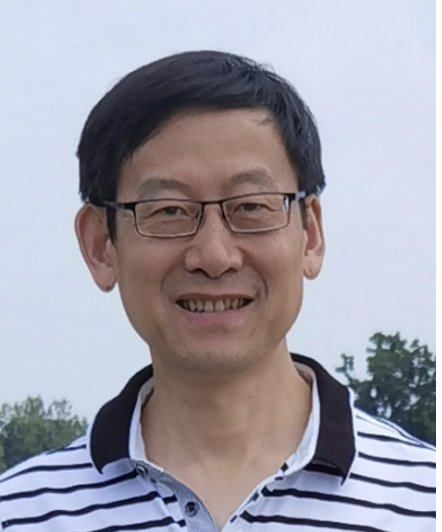}}]{Kun Yang}
received his PhD from the Department of Electronic \& Electrical Engineering of University College London (UCL), UK. He is currently a Chair Professor of Nanjing University and also an affiliated professor at University of Essex and UESTC. His main research interests include wireless networks and communications, communication-computing cooperation, and new AI (artificial intelligence) for wireless. He has published 500+ papers and filed 50 patents. He serves on the editorial boards of a number of IEEE journals (e.g., IEEE WCM, TVT, TNB). He is a Deputy Editor-in-Chief of IET Smart Cities Journal. He has been a Judge of GSMA GLOMO Award at World Mobile Congress – Barcelona since 2019. He was a Distinguished Lecturer of IEEE ComSoc (2020-2021), a Recipient of the 2024 IET Achievement Medals and the Recipient of 2024 IEEE CommSoft TC’s Technical Achievement Award. He is a Member of Academia Europaea (MAE), a Fellow of IEEE, a Fellow of IET and a Distinguished Member of ACM.
\end{IEEEbiography}

\clearpage
\appendices

\section{Proof of Theorem~\ref{thm: average AoC general}}\label{App: average AoC soft}

Recall that $\hat{\epsilon}(t)=\frac{A(t)}{G(t)}$ represents the frequency of task delays exceeding the deadline up to time $t$ and $\{T_k\}_k$ is identical across $k$. The limit of $\hat{\epsilon}_w(t)$ exists and is denoted by  $\epsilon_w$. We have:
\begin{align*}
\lim_{t\to\infty}\hat{\epsilon}_w(t) =\Pr(T_k>w)\triangleq \epsilon_{w}.
\end{align*}
Additionally, $\hat{\epsilon}_w(t)$ remains unchanged until a new task is completed. Consider the $k$-th valid task, where the corresponding inter-arrival time is $X_k$ and the system time is $T_k$. Let the corresponding value of $\hat{\epsilon}_w(t)$ be $\hat{\epsilon}_k$.  We then have:
\begin{align*}
\lim_{k\to\infty}\hat{\epsilon}_k = \lim_{t\to\infty}\hat{\epsilon}_w(t) =\epsilon_{w}.
\end{align*}

To derive an expression for the average AoC, we start with a similar graphical argument from \cite{AoItutorial}. Consider the sum of: (i) the area corresponding to the interval $(\tau_{k-1}, \tau_{k}]$, i.e., the area of trapezoid $\overline{ABCD}$, and (ii) the area corresponding to the additional latency incurred by task $k$, i.e., the area of triangle $\overline{DEF}$ in Fig.~\ref{curveAoCsoft}~(a), or the area of trapezoid $\overline{HGDE}$ in Fig.~\ref{curveAoCsoft}~(b). Note that $S_{\overline{HGDE}} = S_{DEF} - S_{FGH}$. Denote the corresponding sum area as $S_k$, 
based on Fig.~\ref{curveAoCsoft}, we can compute
\begin{align}\label{eq: S sum area}
S_k=&X_kT_k+\frac{X_k^2}{2}+\frac{\hat{\epsilon}_k}{2}\big((T_k-w)^+\big)^2\nonumber\\
-&\frac{\hat{\epsilon}_k}{2}\big((T_k-S_{k, c}-w)^+\big)^2\nonumber\\
\triangleq&Q_{k, 1} + \frac{\hat{\epsilon}_k}{2} Q_{k, 2}.
\end{align}
Since $\{X_k\}_k$, $\{S_{k, c}\}_k$ and $\{T_k\}_k$ are i.i.d, then $\{Q_{k,1}\}_k$ and $\{Q_{k, 2}\}_k$ are also i.i.d, respectively.

Let the number of valid tasks be $K$, and consider the limit as $K\to\infty$.
Utilizing a similar graphical argument in Appendix~\ref{App: average AoC soft}, the average AoC can be computed by
\begin{align*}
\hat{\Theta}_{\text{soft}} =& \lim_{K\to\infty}\frac{\sum_{k=1}^K\Big(Q_{k, 1} + \frac{\hat{\epsilon}_k}{2} Q_{k, 2}\Big)}{\sum_{k=1}^{K}X_k}\\
=&\lim_{K\to\infty}\frac{\frac{1}{K}\sum_{k=1}^K\Big(Q_{k, 1} + \frac{\hat{\epsilon}_k}{2} Q_{k, 2}\Big)}{\frac{1}{K}\sum_{k=1}^{K}X_k}.
\end{align*}
By the Law of Large Numbers, $\Theta_{\text{soft}}$ reduces to 
\begin{align}\label{eq: hatTheta-1}
\Theta_{\text{soft}} =&\frac{\mathbb{E}[Q_{k, 1}]}{\mathbb{E}[X_k]}+\lim_{K\to\infty}\frac{\frac{1}{K}\sum_{k=1}^K\frac{\hat{\epsilon}_k}{2}Q_{k, 2}}{\mathbb{E}[X_k]}.
\end{align}

Now, let us focus on the term $\frac{1}{K}\sum_{k=1}^KQ_{k, 2}$. Since $\lim_{k\to\infty}\hat{\epsilon}_k=\epsilon_w$, for any small $\eta>0$, there exists a large integer $H(\eta)$ such that $|\hat{\epsilon}_k - \epsilon_w|\leq\eta$. Then:
\begin{align*}
\lim_{K\to\infty}\frac{1}{K}\sum_{k=1}^K\frac{\hat{\epsilon}_k}{2}Q_{k,2} =& \lim_{K\to\infty}\frac{1}{K}\big(\sum_{k=1}^{H(\eta)}+\sum_{k=H(\eta)+1}^K\big)\frac{\hat{\epsilon}_k}{2}Q_{k,2} \\
\overset{(a)}{=}&\lim_{K\to\infty}\frac{1}{K}\sum_{k=1}^{H(\eta)}\frac{\hat{\epsilon}_k}{2}Q_{k,2}.
\end{align*}
The equality (a) holds because $$\lim_{K\to\infty}\frac{1}{K}\sum_{k=H(\eta)+1}^{K}\frac{\hat{\epsilon}_k}{2}Q_{k,2}=0,$$ as there are only finite number ($H(\eta)$) terms and each term is a finite scalar.  

Since  $|\hat{\epsilon}_k - \epsilon_w|\leq\eta$, the summation $\frac{1}{K}\sum_{k=1}^K\frac{\hat{\epsilon}_k}{2}Q_{k, 2}$ is bounded by
\begin{align*}
	&\lim_{K\to\infty}\frac{\epsilon_w-\eta}{2K}\sum_{k=H(\eta)}^KQ_{k, 2}
	\leq\lim_{K\to\infty}\frac{1}{K}\sum_{k=H(\eta)}^K\frac{\hat{\epsilon}_k}{2}Q_{k, 2}\\
	&\leq\lim_{K\to\infty}\frac{\epsilon_w+\eta}{2K}\sum_{k=H(\eta)}^KQ_{k, 2}.
\end{align*}
By the Law of Large Numbers, we have:
\begin{align}\label{eq: hatTheta-2}
&\frac{\epsilon_w-\eta}{2}\mathbb{E}[Q_{k, 2}]\leq\lim_{K\to\infty}\frac{1}{K}\sum_{k=H(\eta)}^K\frac{\hat{\epsilon}_k}{2}Q_{k, 2}\nonumber\\ &\leq\frac{\epsilon_w+\eta}{2} \mathbb{E}[Q_{k, 2}].
\end{align}
Substituting \eqref{eq: S sum area} and \eqref{eq: hatTheta-2} into \eqref{eq: hatTheta-1}, and taking $\eta\to0$, we obtain \eqref{eq: average AoC soft}.

\section{The proof of Theorem~\ref{thm: average AoC soft}}\label{App: average AoC soft closedform}
First of all, we denote
\begin{align*}
&Q_{1} = \frac{\mathbb{E}[X_kT_k+\frac{X_k^2}{2}]}{\mathbb{E}[X_k]},\,\,Q_{2} = \epsilon_w\cdot\frac{q_1 - q_2}{2\mathbb{E}[X_k]}\\
&q_1 = \mathbb{E}\big[\big((T_k-w)^+\big)^2\big]\\
&q_2 =\mathbb{E}\big[\big((T_k-S_{k,c}-w)^+\big)^2\big].
\end{align*}
From \cite[Proposition~1]{AoI_edgecomputing}, $Q_{1}$ has the same formula of the average AoI, which has the following expression,
\begin{align}\label{eq: AoC-AoI}
	&Q_{1}=\frac{1}{\lambda}+\frac{1}{\mu_t}+\frac{1}{\mu_c}+\frac{\lambda^2}{\mu_t^2(\mu_t-\lambda)}\nonumber\\
	&+ \frac{\lambda^2}{\mu_c^2(\mu_c-\lambda)}+\frac{\lambda^2}{\mu_t\mu_c(\mu_t+\mu_c-\lambda)}.
\end{align}
To obtain the closed expression for $\Theta_{\text{soft}}$ in \eqref{eq: average AoC soft}, it suffices to obtain the closed expression for $Q_2$.

{\bf Step 1}. We first obtain $\epsilon_w$.
Denote $U_{t}$ and $U_{c}$ are the system time in the transmission queue and the computation queue, respectively. Since the tandem is a combination of M/M/1 queues. Due to the memoryless property, $U_{t}$ and $U_{c}$ are independent \cite{QueueingBook, AoI_edgecomputing}. The density functions for $U_{t}$ and $U_{c}$ are given by \cite{QueueingBook}:
\begin{align}
f_{U_t}(x) =& (\mu_t-\lambda)e^{-(\mu_t-\lambda)x}\\
f_{U_c}(x) =& (\mu_c-\lambda)e^{-(\mu_c-\lambda)x}.
\end{align}
Due to independence, the densifty function of the total system time, i.e., $T_k=U_t+U_c$, is calculated by convolution \cite{QueueingBook}:
\begin{align}\label{eq: density Tk}
	f_{T_k}(x)=\left\{
	\begin{aligned}
		&\frac{\mu_t\delta_t\mu_c\delta_c}{\mu_c-\mu_t}(e^{-\mu_t\delta_tx}-e^{-\mu_c\delta_cx}),\,\,\mu_t\neq\mu_c\\
		&\mu_t^2\delta_t^2 xe^{-\mu_t\delta_tx},\,\,\mu_t=\mu_c.
	\end{aligned}
	\right.
\end{align}
Recall that $\epsilon_{w}=\Pr\big(T_k>w\big)$. From \eqref{eq: density Tk}, $\epsilon_{w}$ can be computed as
\begin{align}\label{eq: prob exceeding-1}
	\epsilon_{w} = \left\{
	\begin{aligned}
		&\frac{\mu_c\delta_c e^{-\mu_t\delta_t w}-\mu_t\delta_te^{-\mu_c\delta_c w}}{\mu_c-\mu_t}, &&\mu_t\neq \mu_c\\
		&(1+\mu_t\delta_tw)e^{-\mu_t\delta_t w}, &&\mu_t=\mu_c.
	\end{aligned}
	\right.
\end{align}

{\bf Step 2}. We compute $q_1$. 

For brevity, we denote 
\begin{align*}
\zeta_t = e^{-\mu_t\delta_tw},\,\,\zeta_c=e^{-\mu_c\delta_cw}.
\end{align*}
We first consider the case when $\mu_t\neq\mu_c$, we then consider the case when $\mu_t = \mu_c$. 

When $\mu_t\neq\mu_c$, according to \eqref{eq: density Tk},  we have
\begin{align*}
	\mathbb{E}\big[\big((T_k-w)^+\big)^2\big] =& \frac{\mu_t\delta_t\mu_c\delta_c}{\mu_c-\mu_t}\Big(\int_{w}^{\infty}(x-w)^2e^{-\mu_t\delta_t x}dx\\
	-&\int_{w}^{\infty}(x-w)^2e^{-\mu_c\delta_c x}dx\Big).
\end{align*}
By some algebra,
\begin{align*}
	&\int_{w}^{\infty}(x-w)^2e^{-\mu_t\delta_t x}dx=\frac{2\zeta_t }{\mu_t^3\delta_t^3}\\
	&\int_{w}^{\infty}(x-w)^2e^{-\mu_c\delta_c x}dx=\frac{2\zeta_c}{\mu_c^3\delta_c^3}.
\end{align*}
Therefore, when $\mu_t\neq \mu_c$,
\begin{align}\label{eq: Etw-1}
	\mathbb{E}\big[\big((T_k-w)^+\big)^2\big] =\frac{\mu_t\delta_t\mu_c\delta_c}{\mu_c-\mu_t}\Big(\frac{2\zeta_t }{\mu_t^3\delta_t^3} - \frac{2\zeta_c}{\mu_c^3\delta_c^3}\Big).
\end{align}
Next, when $\mu_t=\mu_c$, according to \eqref{eq: density Tk}, we have
\begin{align}\label{eq: Etw-2}
\mathbb{E}\big[\big((T_k-w)^+\big)^2\big]=\zeta_t (\frac{6}{\mu_t^2\delta_t^2}+\frac{2w}{\mu_t\delta_t}).
\end{align}

{\bf Step 3}. We compute $q_2$ by considering the waiting time in the computation queue, denoted as $W$. We have the following relation:
\begin{align*}
T_k - S_{k, c} = U_t+W.
\end{align*}
According to \cite{QueueingBook}, the waiting time in the computation queue, $W$, follows the density function:
\begin{align*}
f_{W}(x) = \rho_c\mu_c(1-\rho_c) e^{-\mu_c(1-\rho_c)x} + (1-\rho_c)\delta(x),
\end{align*}
where $\delta(x)$ is the Dirac delta function.
Since $U_t$ and $U_c$ are independent, $W$ and $U_t$ are also independent. By convolution, we can derive the density function of $U_t + W$. When $\mu_t\neq \mu_c$, we have:
\begin{align*}
f_{U_t+W}(x) =& (1-\rho_c)\mu_t\delta_te^{-\mu_t\delta_t x} \\
+& \frac{\rho_c\mu_c\delta_c\mu_t\delta_t(e^{-\mu_t\delta_t x}-e^{-\mu_c\delta_c x})}{\mu_c-\mu_t},
\end{align*}
and
\begin{align}\label{eq: q2-1}
q_2 =\frac{2\zeta_t}{\mu_t^3\delta_t^3}\frac{\mu_c\delta_c\mu_t\delta_t-\delta_c\mu_t^2\delta_t^2}{\mu_c-\mu_t}- \frac{2\zeta_c}{\mu_c^3\delta_c^3}\frac{\rho_c\mu_c\delta_c\mu_t\delta_t}{\mu_c-\mu_t}.
\end{align}
When $\mu_t=\mu_c$, we have:
\begin{align*}
f_{U_t+W}(x) =\mu_t\delta_t^2e^{-\mu_t\delta_t x} +\rho_t\mu_t^2\delta_t^2 t e^{-\mu_t\delta_t x},
\end{align*}
and
\begin{align}\label{eq: q2-2}
q_2 =\frac{2\zeta_t}{\mu_t^2\delta_t}+\rho_t\zeta_t( \frac{2w}{\mu_t\delta_t}+\frac{6}{\mu_t^2\delta_t^2}).
\end{align}

From {\bf Step 1}, {\bf Step 2} and {\bf Step 3}, substituting \eqref{eq: AoC-AoI}, \eqref{eq: prob exceeding-1}, \eqref{eq: Etw-1}, \eqref{eq: Etw-2}, \eqref{eq: q2-1}, and \eqref{eq: q2-2} into \eqref{eq: average AoC soft}, we have the desired results.

	\section{Proof of Theorem~\ref{thm: average AoC soft extension}}\label{App: proof of average AoC soft extension}
	The proof is inspired by the proof of \cite[Proposition~1]{AoI_edgecomputing}.
	
	We first divide the delay of task $k$ into $4$ components: 
	\begin{align*}
		T_k = W_{k, t} + S_{k, t} + W_{k, c} + S_{k, c},
	\end{align*}
	where $W_{k, t}=U_{k, t}-S_{k, t}$ and $W_{k, c}=U_{k, c}-S_{k,c}$ denote the waiting times of task $k$ in the transmission and computation queues, respectively, while $S_{k, t}$ and $S_{k, c}$ represent the service times of task $k$ in the transmission and computation queues, respectively.  From \eqref{eq: average AoC soft}, 
	$\Theta_{\text{soft}}$ can be re-written as
	\begin{align}\label{eq: Theta soft extend}
		\Theta_{\text{soft}} = & \frac{1}{\mathbb{E}[X_k]}\Big(\mathbb{E}[X_kW_{k, t}] + \mathbb{E}[X_kS_{k, t}]  \nonumber  \\
		+& \mathbb{E}[X_kW_{k, c}] + \mathbb{E}[X_kS_{k, c}] + \frac{1}{2}\mathbb{E}[X_k^2]\nonumber\\
		+&\frac{\Pr(T_k>w)}{2}\mathbb{E}[\big((T_k-w)^+\big)^2]\nonumber\\
		-&\frac{\Pr(T_k>w)}{2}\mathbb{E}[\big((T_k-S_{k,c}-w)^+\big)^2]\Big).
	\end{align}
	Note that $X_k$ is independent of $S_{k, t}$ and $S_{k, c}$, with $\mathbb{E}[X_k] = \frac{1}{\lambda}$, $\mathbb{E}[S_{k, t}] = \frac{1}{\mu_t}$, and $\mathbb{E}[S_{k, c}] = \frac{1}{\mu_c}$. Therefore, to compute $\Theta_{\text{soft}}$, it suffices to determine $\mathbb{E}[X_kW_{k, t}]$, $\mathbb{E}[X_kW_{k, c}]$, and $\frac{\Pr(T_k>w)}{2}\mathbb{E}[\big((T_k-w)^+\big)^2]$. Denote $g_1 = \mathbb{E}[X_kW_{k, t}]$, $g_2 = \mathbb{E}[X_kW_{k, c}]$,  $g_3 = \frac{\Pr(T_k>w)}{2}$, $g_4=\mathbb{E}[\big((T_k-w)^+\big)^2]$, and $g_5=\mathbb{E}[\big((T_k-S_{k, c}-w)^+\big)^2]$.

	{\bf Step 1}. To obtain $g_1$, we proceed as follows:
	\begin{align}\label{eq: g1-1}
		g_1 = \mathbb{E}[X_kW_{k, t}] = \int_{0}^{\infty} x \mathbb{E}[W_{k, t}|X_k = x]f_X(x)dx.
	\end{align}
	Since $W_{k, t}$ represents the waiting time in the transmission queue, we have $W_{k, t} = \big(U_{k-1, t} - X_k\big)^+$. Thus,
	\begin{align}
		\mathbb{E}[W_{k, t}|X_k = x] =& \mathbb{E}[\big(U_{k-1, t}- X_k\big)^+|X_k = x]\nonumber\\
		=&\mathbb{E}[(U_{k-1, t}-x)^+]\nonumber \\
		=&\int_{x}^{\infty}(\tau-x)f_{U_t}(\tau)d\tau.\label{eq: expectation W|X}
	\end{align}
	Substituting \eqref{eq: expectation W|X} into \eqref{eq: g1-1}, we obtain:
	\begin{align*}
		g_1 = \int_{0}^{\infty} x \int_{x}^{\infty} (\tau-x)f_{U_t}(\tau)d\tau f_X(x)dx.
	\end{align*}

	{\bf Step 2}. To obtain $g_2$, we proceed as follows:
	\begin{align}\label{eq: g2-1}
		g_2 =& \mathbb{E}[X_kW_{k, c}] = \int_{0}^{\infty} x \mathbb{E}[W_{k, c}|X_k = x]f_X(x)dx\nonumber\\
		=&\int_{0}^{\infty}x\Big(\int_{0}^{\infty}\tau f_{W_{k, c}|X_k = x}(\tau)d\tau\Big)f_X(x)dx.
	\end{align}
	The rest of the step is to find $f_{W_{k, c}|X_k = x}(\tau)d\tau$.
	
	We introduce another random variable, $D_{k, t}$, which represents the {\it inter-departure time} associated with task $k$ in the transmission queue.  Note that $X_k-D_{k, t}-W_{k, c}$ forms a Markov chain. By the chain rule,
	\begin{align*}
		&\Pr\{X_k=x, D_{k, t}=y, W_{k, c}=a\} = \Pr\{X_k=x\} \\
		&\times\Pr\{D_{k, t}=y|X_{k}=x\} \Pr\{W_{k, c}=a|D_{k, t}=y\},
	\end{align*}
	which implies
	\begin{align*}
		&\frac{\Pr\{X_k=x, D_{k, t}=y, W_{k, c}=a\}}{\Pr\{X_k=x\}} \\
		&= \Pr\{D_{k, t}=y|X_{k}=x\} \Pr\{W_{k, c}=a|D_{k, t}=y\},
	\end{align*}
	and therefore,
	\begin{align*}
		&\Pr\{D_{k, t}=y, W_{k, c}=a|X_k=x\}\\
		& = \Pr\{D_{k, t}=y|X_{k}=x\} \Pr\{W_{k, c}=a|D_{k, t}=y\}.
	\end{align*}
	Using the definition of density functions, we replace the probabilities with their respective density functions. Integrating with respect to $D_{k, t}$ on both sides, we have:
	\begin{align}\label{eq: f(W|X)}
		f_{W_{k, c}|X_{k}=x}(\tau) = \int_{0}^{\infty}f_{D_{k, t}|X_k=x}(y)f_{W_{k, c}|D_{k, t}=y}(\tau)dy.
	\end{align}
	To obtain \eqref{eq: f(W|X)}, we need to obtain $f_{D_{k, t}|X_k=x}(y)$ and $f_{W_{k, c}|D_{k, t}=y}(\tau)$. 
	
	Recall $W_{k, c}$ is the waiting time of task $k$ in the computation queue, we have $W_{k, c} = (U_{k-1, c} - D_{k, t})^+$. Thus,
	\begin{align}\label{eq: f(D|X)}
		& f_{W_{k, c}|D_{k, t}=y}(\tau) =  f_{(U_{k-1, c} - D_{k, t})^+|D_{k, t}=y}(\tau)\nonumber\\
		&= f_{(U_{k-1, c} - y)^+}(\tau) = f_{U_c}(\tau+y).
	\end{align}
	
	At the end of this step, we obtain $f_{D_{k, t}|X_k=x}(y)$. Define $P_{\text{busy}, k}$ and $P_{\text{idle}, k}$ as the probabilities that the transmission queue is busy or idle, respectively, {\it upon the arrival} of task $k$. Then:
	\begin{align}\label{eq: f(D|X)-0}
		f_{D_{k, t}|X_k=x}(y) =& P_{\text{busy}, k}\cdot f_{D_{k, t}|W_{k, t}>0, X_k=x}(y) \nonumber\\
		+& P_{\text{idle}, k}\cdot f_{D_{k, t}|W_{k, t}=0, X_k=x}(y).
	\end{align}
	
	When $W_{k, t}>0$, the inter-departure time of task $k$ and task $k-1$ in the transmission queue equals the service time of task $k$. Thus:
	\begin{align}
		&f_{D_{k, t}|W_{k, t}>0, X_k=x}(y) = f_{S_{t}}(y)\label{eq: P_{busy}-1}\\
		&P_{\text{busy}, k} = \Pr(W_{k, t}>0|X_k=x)\nonumber\\ 
		&=\Pr(U_{k-1, t}>x) = \int_{x}^{\infty} f_{U_t}(\tau)d\tau.\label{eq: P_{busy}-2}
	\end{align}
	
	When $W_{k, t}=0$, $D_{k, t}$ consists of $2$ components: the service time of task $k$, $S_{k, t}$, and the interval between the arrival of task $k$ and the departure of task $k-1$, denoted by $B_{k, t}$. Hence:
	\begin{align}\label{eq: f(D|W, X)}
		f_{D_{k, t}|W_{k, t}=0, X_k=x}(y) = f_{S_t}(y) \ast f_{B_{k, t}|W_{k, t}=0, X_k=x}(y),
	\end{align}
	where $\ast$ represents the convolution operation. 
	
	It follows that
	\begin{align}\label{eq: P(idle)}
		P_{\text{idle}, k} =& \Pr(W_{k, t}=0|X_k=x)=\Pr(U_{k-1, t}<x)\nonumber\\
		=&\int_{0}^{x}f_{U_t}(\tau)d\tau.
	\end{align}
	From the definition of $B_{k, t}$ as $B_{k, t} = X_k - U_{k-1, t}$, and using the result from \cite[Appendix~A]{AoI_edgecomputing}, we obtain:
	\begin{align}\label{eq: f(B)}
		&f_{B_{k, t}|W_{k, t}=0, X_k=x}(y) = \frac{f_{B_{k, t}| X_k=x}(y)}{P_{\text{idle}, k}}\nonumber\\
		=& \frac{f_{X_k - U_{k-1, t} | X_k=x}(y)}{P_{\text{idle}, k}} = \frac{f_{U_{k-1, t} | X_k=x}(X_k - y)}{P_{\text{idle}, k}}\nonumber\\
		=&\frac{f_{U_{k-1, t}}(x - y)}{P_{\text{idle}, k}}
	\end{align}
	Substituting \eqref{eq: f(D|W, X)} and \eqref{eq: f(B)} into \eqref{eq: f(D|W, X)}, we get:
	\begin{align}\label{eq: f(D|W, X)-1}
		f_{D_{k, t}|W_{k, t}=0, X_k=x}(y) = f_{S_t}(y) \ast \frac{f_{U_{k-1, t}}(x - y)}{\int_{0}^{x}f_{U_t}(\tau)d\tau}.
	\end{align}
	
Let $\xi(x) = \int_{x}^{\infty} f_{U_t}(u)du$.	Using \eqref{eq: g2-1}, \eqref{eq: f(W|X)}, \eqref{eq: f(D|X)-0}, \eqref{eq: P_{busy}-1}, \eqref{eq: P_{busy}-2}, \eqref{eq: P(idle)}, and \eqref{eq: f(D|W, X)-1}, we derive:
	\begin{align*}
		g_2 =& \int_{0}^{\infty}x \int_{0}^{\infty} \tau 
		\int_{0}^{\infty} (\xi(x)\cdot f_{S_{t}}(y)+ f_{S_t}(y) \ast f_{U_{t}}(x - y)\big)\\
		&\cdot
		f_{U_c}(\tau+y)dy
		d\tau f_X(x)dx.
	\end{align*}
	
	{\bf Step 3}. To obtain $g_3$, $g_4$, and $g_5$. Note that $T_k = U_{k, t}+ U_{k, c}$. Then:
	\begin{align*}
		f_{T_k}(\tau) = \int_{0}^{\tau}f_{U_t, U_c}(u, \tau-u)du.
	\end{align*}
	Substituting $f_{T_k}(\tau)$ into $g_3$, $g_4$, and $g_5$, we have
	\begin{align*}
		&g_3 = \int_{w}^{\infty}\int_{0}^{\tau}f_{U_t, U_c}(u, \tau-u)dud\tau\\
		&g_4=\int_{w}^{\infty}(\tau-w)^2\cdot\int_{0}^{\tau}f_{U_t, U_c}(u, \tau-u)dud\tau\\
		&g_5 = \int_{w}^{\infty}(\tau-w)^2\cdot\int_{0}^{\tau}f_{U_t, U_c-S_c}(u, \tau-u)dud\tau.
	\end{align*}

	From {\bf Steps 1 $\sim$ 3}, we obtain the desired results.

\section{Proof of Theorem~\ref{thm: average AoC hard general}}\label{App: average AoC hard}
Utilizing a similar idea of calculating average AoI \cite{AoItutorial}, the average AoC can be computed as the sum of area of the parallelogram (e.g., $\overline{ABCD}$) over time horizon $t$. The number of corresponding parallelograms is equal to the number of informative tasks. As the time horizon approaches infinity, the rate of informative tasks is
\begin{align*}
\lim_{t\to\infty}\frac{K(t)}{t},
\end{align*}
which is also the rate of parallelograms.
Then, the average AoC is given by
\begin{align}\label{eq: AoC hard proof}
\Theta_{\text{hard}} = \lim_{t\to\infty}\frac{K(t)}{t}\cdot\mathbb{E}[S_{\overline{ABCD}}].
\end{align}
From the distribution of $M$ in \eqref{eq: Mk-1}, we know that a valid task followed by $M-1$ invalid tasks, so 
\begin{align}\label{eq: rate of good tasks}
\lim_{t\to\infty}\frac{K(t)}{t} =\lim_{t\to\infty}\frac{1}{t/K(t)} = \frac{1}{\mathbb{E}[\sum_{j=1}^{M}X_j]}.
\end{align}
In addition, the area of $\overline{ABCD}$ can be calculated by 
\begin{align}\label{eq: areaABCD}
	S_{\overline{ABCD}} =& \mathbb{E}\Big[\big(\sum_{j=1}^{M}X_j+T_M\big)^2/2 - T_M^2/2\Big]\nonumber\\
	=& \mathbb{E}[T_M\cdot\sum_{j=1}^{M}X_j ]+\frac{1}{2}\mathbb{E}[\big(\sum_{j=1}^{M}X_j\big)^2\big].
\end{align}
Substituting \eqref{eq: rate of good tasks} and \eqref{eq: areaABCD} into \eqref{eq: AoC hard proof}, we obtain \eqref{eq: average AoC hard}.

\section{Proof of Theorem~\ref{thm: average AoC hard}}\label{App: average AoC hard closedform}
As discussed before, there are $4$ types of correlations underlying \eqref{eq: average AoC hard}: (i) positive correlations among the delays $\{T_k\}_k$ over $k$, (ii) positive correlations between $T_k$ ($1\leq k\leq M$) and $M$, (iii) negative correlations between $T_k$ and $X_k$, and (iv) negative correlations among $X_k$ and $M$. The correlations in (ii), (iii), and (iv) are incurred by the positive correlations in (i).
We can alleviate these correlations when the positive correlations in (i) are negligible.

In a tandem of two M/M/1 queues with parameters $(\lambda, \mu_t,\mu_c)$, when $\mu_t\gg\lambda$ and $\mu_c\gg\lambda$, the positive correlations among $\{T_k\}_k$ become negligible \cite{QueueingBook}. In other words, $T_k$ and $T_{k+1}$ are approximately independent over $k$. Consequently, due to the approximate independence among $\{T_k\}_k$, according to \eqref{eq: Mk-1}, $M$ approximates a geometric distribution with parameter $1-\epsilon_w$, which is approximately independent of $T_k$. Additionally, when $\mu_t\gg\lambda$ and $\mu_c\gg\lambda$, the delay $T_k$ is dominated by the service times at the transmitter and the computational node. This implies that $T_k$ and $X_k$ are approximately independent. Hence, $X_k$ and $M$ are also approximately independent.

{\bf Step 1}. We prove \eqref{eq: AoC hard-1}. 
Recall that $\{X_k\}_k$ are i.i.d over $k$. As discussed above, when $\mu_t\gg\lambda$ and $\mu_c\gg\lambda$, we have the following: (i) From the model assumption, $\{X_k\}_k$ are independent and identical distributions. Since $M$ is approximately independent of $X_k$,  we have:
\begin{align}
&\mathbb{E}[\sum_{j=1}^{M}X_j] \approx \mathbb{E}[M]\mathbb{E}[X_1]\label{eq: XM-1},\\
&\mathbb{E}[(\sum_{j=1}^{M}X_j)^2]\approx\mathbb{E}[M]\mathbb{E}[X_1^2]\nonumber\\
&+(\mathbb{E}[M^2]-\mathbb{E}[M])\mathbb{E}^2[X_1]\label{eq: XM-2}.
\end{align}
(ii) Since $T_k$ is approximately independent of $X_k$, we have: 
\begin{align}\label{eq: XM-3}
\mathbb{E}[T_M\cdot\sum_{j=1}^{M}X_j] \approx \mathbb{E}[T_M]\mathbb{E}[\sum_{j=1}^{M}X_j ].
\end{align}
Substituting \eqref{eq: XM-1}, \eqref{eq: XM-2}, and \eqref{eq: XM-3} into \eqref{eq: average AoC hard}, we get:
\begin{align*}
\Theta_{\text{hard}} \approx& \mathbb{E}[T_M]+\frac{\mathbb{E}[X_1^2]}{2\mathbb{E}[X_1]}+\big(\frac{\mathbb{E}[M^2]}{2\mathbb{E}[M]}-\frac{1}{2}\big)\mathbb{E}[X_1],
\end{align*}
thereby completing the proof of \eqref{eq: AoC hard-1}.

{\bf Step 2}. We prove \eqref{eq: Thetahard-1}. According to \eqref{eq: Mk-1},
\begin{align}\label{eq: TM}
\mathbb{E}[T_M] = \mathbb{E}[T_1|T_1\leq w].
\end{align}
Since the density function of $T_k$ is given by  \eqref{eq: density Tk}, substituting \eqref{eq: density Tk} into \eqref{eq: TM}, we have
\begin{align}\label{eq: ETM-1}
\mathbb{E}[T_M] = \frac{\frac{1-\frac{1+\mu_t\delta_tw}{e^{\mu_t\delta_tw}}}{\mu_t^2\delta_t^2} - \frac{1-\frac{1+\mu_c\delta_cw}{e^{\mu_c\delta_cw}}}{\mu_c^2\delta_c^2}}{\frac{1-e^{-\mu_t\delta_tw}}{\mu_t\delta_t}-\frac{1-e^{-\mu_c\delta_cw}}{\mu_c\delta_c}}.
\end{align}
Since $X_1$ has an exponential distribution with parameter $\lambda$, we have:
\begin{align}\label{eq: X2X1}
\frac{\mathbb{E}[X_1^2]}{2\mathbb{E}[X_1]} = \frac{1}{\lambda}.
\end{align}
Recall that $M$ approximates a geometric distribution with parameter $1-\epsilon_w$.  Therefore,
\begin{align}\label{eq: M2M1X-1}
\big(\frac{\mathbb{E}[M^2]}{2\mathbb{E}[M]}-\frac{1}{2}\big)\mathbb{E}[X_1] = \frac{1}{\lambda}(\frac{1}{1-\epsilon_w}-1),
\end{align}
where $\epsilon_w$ is given by \eqref{eq: prob exceeding-1}.
Substituting \eqref{eq: prob exceeding-1}, \eqref{eq: ETM-1}, \eqref{eq: X2X1}, and \eqref{eq: M2M1X-1} into \eqref{eq: AoC hard-1}, we obtain \eqref{eq: Thetahard-1}.

{\bf Step 3}. We prove \eqref{eq: Thetahard-2}.
When $\mu_t = \mu_c$, substituting \eqref{eq: density Tk} into \eqref{eq: TM}, we have
\begin{align}\label{eq: ETM-2}
\mathbb{E}[T_M] = \frac{\frac{2}{\mu_t\delta_t}-(\frac{2}{\mu_t\delta_t}+2w+\mu_t\delta_tw^2)e^{-\mu_t\delta_tw} }{1-e^{-\mu_t\delta_tw}(1+\mu_t\delta_tw)}.
\end{align}
Substituting \eqref{eq: prob exceeding-1}, \eqref{eq: ETM-1},  \eqref{eq: X2X1},  \eqref{eq: M2M1X-1}, and \eqref{eq: ETM-2} into \eqref{eq: AoC hard-1}, we obtain \eqref{eq: Thetahard-2}.

\section{Proof of Lemma~\ref{lem: computation throughput}}\label{App: computation throughput}
Consider a large time horizon $T$, during which there are $K(T)$ informative packets. For each informative task $j\in\{1,2,\cdots,K(T)\}$, denote the number of the associated bad tasks as $M_j$. $\{M_j\}_j$ is identical over $j$ and follows the distribution given in \eqref{eq: Mk-1}. Since $\{X_k\}_k$ are i.i.d over $k$, the sequence  $\{\sum_{k=1}^{M_j}X_{k}\}_j$ are identical over $j$. The remaining time in the interval $[0, T]$ is $T-\tau'_{K(T)}$. Thus, the time horizon can be re-written as
\begin{align}\label{eq: time horizon}
T = \sum_{j=1}^{K(T)}\sum_{k=1}^{M_j}X_{k+j-1}+T-\tau'_{K(T)}.
\end{align}
Substituting \eqref{eq: time horizon} into \eqref{eq: computation throughput}, we have
\begin{align*}
&\Xi = \lim_{T\to\infty}\frac{K(T)}{T} \\
&= \lim_{T\to\infty}\frac{K(T)}{\sum_{j=1}^{K(T)}\sum_{k=1}^{M_j}X_{k+j-1}+T-\tau'_{K(T)}}\\
&=\lim_{T\to\infty}\frac{1}{\frac{1}{K(T)}\sum_{j=1}^{K(T)}\sum_{k=1}^{M_j}X_{k+j-1}+\frac{T-\tau'_{K(T)}}{K(T)}}.
\end{align*}
Since the sequence $\{\sum_{k=1}^{M_j}X_{k}\}_j$ is identical over $j$, by the central limit theory, we have
\begin{align*}
\Xi=&\frac{1}{\mathbb{E}[\sum_{k=1}^{M}X_k]+\lim_{T\to\infty}\frac{T-\tau'_{K(T)}}{K(T)}}\\
=&\frac{1}{\mathbb{E}[\sum_{k=1}^{M}X_k]}.
\end{align*}

\section{Proof of Proposition~\ref{pro: approximated CT}}\label{App: approximated CT}
Recall that $M$ approximates a geometric distribution with parameter $1-\epsilon_w$, and $X$ follows an exponential distribution with parameter $\lambda$,  then
\begin{align}\label{eq: EM}
\frac{1}{\mathbb{E}[M]\mathbb{E}[X_1]} = \lambda(1-\epsilon_w).
\end{align}
Subsituting \eqref{eq: prob exceeding-1} into \eqref{eq: EM}, we obtain \eqref{eq: TC-1} and \eqref{eq: TC-2}.

\section{Proof of Lemma~\ref{lem: pareto optimal}}\label{App: lem: pareto optimal}
The proof is based on contradiction. Assume that the pair $\big(\hat{\Theta}(u), \hat{\Xi}(u)\big)$ is not a weakly Pareto-optimal point. This implies that there exists another solution $\big(\hat{\Theta}'(u), \hat{\Xi}'(u)\big)$ such that $\hat{\Theta}'(u) < \hat{\Theta}(u)$ and $\hat{\Xi}'(u) > \hat{\Xi}(u)$. Given that $\hat{\Xi}'(u) > \hat{\Xi}(u)>u$, the solution $\big(\hat{\Theta}'(u), \hat{\Xi}(u)'\big)$ must be a feasible solution to problem \eqref{eq: pareto optimization}. However, since $\hat{\Theta}(u)$ is the minimum value in the problem \eqref{eq: pareto optimization}, it follows that $\hat{\Theta}(u)\leq \hat{\Theta}'(u)$. This contradicts the assumption that $\hat{\Theta}'(u) < \hat{\Theta}(u)$.

	\section{Proof of Theorem~\ref{thm: average AoC hard extension}}\label{App: proof of average AoC hard extension}
	
	In a tandem of two G/G/1 queues with parameters $(\lambda, \mu_t,\mu_c)$, when $\mu_t\gg\lambda$ and $\mu_c\gg\lambda$, the positive correlations among $\{T_k\}_k$ become negligible \cite{QueueingBook}. In other words, $T_k$ and $T_{k+1}$ are approximately independent over $k$. Consequently, due to the approximate independence among $\{T_k\}_k$, and according to \eqref{eq: Mk-1}, $M$ approximates a geometric distribution with parameter $1-\epsilon_w$, which is approximately independent of $T_k$. Additionally, when $\mu_t\gg\lambda$ and $\mu_c\gg\lambda$, the delay $T_k$ is dominated by the service times at the transmitter and the computational node. This implies that $T_k$ and $X_k$ are approximately independent. Hence, $X_k$ and $M$ are also approximately independent. Based on the same process in {\bf Step 1} in Appendix~\ref{App: average AoC hard closedform}, we show that the average AoC defined in \eqref{eq: average AoC hard} can be accurately approximated by \eqref{eq: AoC hard-1}.

	{\bf Step 1}. The delay of task $k$, $T_k$, is the sum of the system times in both transmission and computation queues, i.e., $T_k = U_{k, t} + U_{k, c}$.
	Based on Assumption~\ref{ass: distributions}, the density functions $f_{U_t}$ and $f_{U_c}$ are known. Denote the density function and CDF of $T_k$ as follows,
	\begin{align*}
		f_{T_k}(\tau)=& \int_{0}^{\tau}f_{U_t,U_c}(u, \tau-u)du\\
		F_{T_k}(x) = &\int_{0}^{x}f_T(\tau)d\tau.
	\end{align*} 
	When $\mu_t\gg\lambda$ and $\mu_c\gg\lambda$, the sequence $\{T_k\}_k$ is approximately independent. Therefore, from \eqref{eq: Mk-1}, we have
	\begin{align}\label{eq: prof of M=n}
		\Pr(M=n) \approx& \big(\Pr(T_k>w)\big)^{n-1}\Pr(T_k\leq w)\nonumber\\
		=& \big(1-F_T(w)\big)^{n-1}F_T(w).
	\end{align}
	Based on \eqref{eq: prof of M=n}, it is straightforward to calculate,
	\begin{align}
		\mathbb{E}[M] =& \frac{1}{F_T(w)}\label{eq: App E[M]}\\
		\mathbb{E}[M^2] =&\frac{2-F_T(w)}{F_T^2(w)}\label{eq: App E[M2]}.
	\end{align}

	{\bf Step 2}. The random variable $T_M$ is the value of $T_k$ at the stopping time $M$, which is the first instance when $T_M\leq w$. Since $\{T_k\}_k$ are approximately i.i.d over $k$, we have 
	\begin{align}\label{eq: App E[TM]}
		\mathbb{E}[T_M] = \mathbb{E}[T_k|T_k\leq w] = \frac{\int_{0}^{w}\tau f_T(\tau)d\tau}{F_T(w)}.
	\end{align}
	Substituting \eqref{eq: App E[M]}, \eqref{eq: App E[M2]}, and \eqref{eq: App E[TM]} into \eqref{eq: AoC hard-1}, we obtain \eqref{eq: AoC hard-extension}.

	\section{Proof of Proposition~\ref{pro: approximated CT extension}}\label{App: approximated CT extension}
	From \eqref{eq: prof of M=n} in Appendix~\ref{App: proof of average AoC hard extension}, $M$ approximates a geometric distribution with parameter $F_T(w)$. Therefore, we have
	\begin{align*}
		\frac{1}{\mathbb{E}[M]\mathbb{E}[X_1]} =\frac{F_T(w)}{\int_{0}^{\infty}xf_X(x)dx},
	\end{align*}
	which derives the desired result.

\section{Proof of \eqref{eq: recursion c soft-1}}\label{App: csoftt+2}
	
We first derive the recursion of $c^{(i)}_{\text{soft}}(k)$. If a task  from source~$i$ leaves the computational node at time $k+1$, i.e., $d_i(k+1)=1$, this task must reach the receiver at the beginning of the current time slot due to $\mu_c=1$.
The instantaneous delay of transmitter $i$ at the end of time $k$ is $z_i(k)$, then from \eqref{eq: AoC-soft}, we have
\begin{align*}
c^{(i)}_{\text{soft}}(k+1) =& z_i(k) + 1 +  \frac{A_i(k)}{G_i(k)}(z_i(k)+1-w)^+,
\end{align*}
where $G_i(k)$ is defined in \eqref{eq: index of information packet} and $A_i(k)$ is defined in \eqref{eq: index of delay not exceeding w}. Otherwise, if no task from source~$i$ leaves the computational node at time $k+1$, we have:
\begin{align*}
c^{(i)}_{\text{soft}}(k+1) = c^{(i)}_{\text{soft}}(k) + 1.
	\end{align*}
	Thus, the recursion for $c^{(i)}_{\text{soft}}(k)$ is given by: 
	\begin{align}\label{eq: recursion c soft}
		&c^{(i)}_{\text{soft}}(k+1) = 1_{\{d_i(k+1)=0\}}\big(c^{(i)}_{\text{soft}}(k) + 1\big)\nonumber\\
		&+1_{\{d_i(k+1)=1\}}\big(z_i(k) + 1 + \ell_i\big(z_i(k)\big),
	\end{align}
	where $\ell_i\big(z_i(k)\big)=\frac{A_i(k)}{G_i(k)}(z_i(k)+1-w)^+$.

Note that a task takes at least $2$ slots from starting transmission to leaving the computational node. Next, we incorporate $\{a_i(k)\}_{i=1}^N$ into the relationship between  $c^{(i)}_{\text{soft}}(k+1)$ and $c^{(i)}_{\text{soft}}(k)$:
\begin{itemize}
\item If $\sum_{i=1}^{N} d_i(k-1) = 0$, meaning no scheduling decision will be made in time slot $k$, so no task can leave the computational node at time $k+1$. In this case, from  \eqref{eq: recursion c soft}, we have  $c^{(i)}_{\text{soft}}(k+1) = c^{(i)}_{\text{soft}}(k) + 1$.
\item If $\sum_{i=1}^{N} d_i(k-1) = 1$ and $a_i(k) = 0$, meaning a task from transmitter $i$ can not leave the computational node at time $k+1$. In this case, from \eqref{eq: recursion c soft}, we have $c^{(i)}_{\text{soft}}(k+1) = c^{(i)}_{\text{soft}}(k) + 1$.
\item  If $\sum_{i=1}^{N} d_i(k-1) = 1$ and $a_i(k) = 1$, meaning transmitter $i$ is scheduled in time slot $k$, then:
\begin{itemize}
			\item If the transmission is completed in one time slot, the current task can leave the computational node at the end of time $k+1$, i.e., $d_i(k+1)=1$. In this case, from \eqref{eq: recursion c soft}, we have $c_i^{\text{soft}}(k+1) = z_i(k) + 1 + \ell_i\left(z_i(k)\right)$.
			\item If the transmission is not completed in one time slot, then from \eqref{eq: recursion c soft}, we have $c_i^{\text{soft}}(k+1) = c_i^{\text{soft}}(k) + 1$.
		\end{itemize}
		
	\end{itemize}
	
	Based on the discussion above, we can write the expression for $c_i^{\text{soft}}(k+1)$ as follows,
	\begin{align*}
		&c_i^{\text{soft}}(k+1) = 1_{\{\sum_{i=1}^{N}d_i(k-1)=0\}}\big(c_i^{\text{soft}}(k) + 1\big)\\
		&+1_{\{\sum_{i=1}^{N}d_i(k-1)=1, a_i(k)=0\}}\big(c_i^{\text{soft}}(k) + 1\big)\\
		&+1_{\{\sum_{i=1}^{N}d_i(k-1)=1, a_i(k)=1, d_i(k+1)=0\}}\big(c_i^{\text{soft}}(k) + 1\big)\\
		&+1_{\{\sum_{i=1}^{N}d_i(k)=1, a_i(k+1)=1, d_i(k+1)=1\}}\\
		&\cdot\Big(z_i(k) + 1 + \ell_i\big(z_i(k)\big)\Big).
	\end{align*}

	\section{Proof of Propositioni~\ref{pro: maxweight soft}}\label{App: maxweight soft}
	Note that given $a_i(k)=1$, for all $k$, we have:
	\begin{align}\label{eq: prob d}
	\Pr(d_i(k+1)=1|a_i(k)=1)=\mu_{t, i}.
	\end{align}
Substituting \eqref{eq: prob d} into \eqref{eq: Drift} and \eqref{eq: recursion c soft-1}, we calculate the drift as follows:
	\begin{align*}
		\Delta(\mathcal{S}(k)) =& \mathbb{E}[L(k+1) - L(k-1) |\mathcal{S}(k)] \nonumber \\
		=& \frac{2}{N}\sum_{i=1}^{N}\beta_i - \frac{1}{N}\sum_{i=1}^{N}\beta_i\mu_{t, i}\mathbb{E}[a_i(k)|\mathcal{S}(k)]w_i(k),
	\end{align*}
	where $w_i(k)$ is given in Algorithm~\ref{alg: soft}. To minimize the drift $\Delta(\mathcal{S}(k))$, Algorithm~\ref{alg: soft} selects the transmitter with the maximum weight $w_i(k)$.

\end{document}